\documentclass[conference]{IEEEtran}

\usepackage{cite}

\ifCLASSINFOpdf
\else
\fi

\usepackage{url}

\usepackage{graphics}
\usepackage{proof}
\usepackage{wrapfig}
\usepackage{listings}
\usepackage{stmaryrd}
\usepackage{amssymb,amsthm,amsmath,amsfonts,mathtools}
\usepackage{color}
\usepackage{hyperref}

\usepackage{pdfpages}
\usepackage{graphicx} 
\usepackage{caption}
\usepackage{subcaption}

\usepackage{tikzit}

\tikzstyle{observation}=[text=blue]
\tikzstyle{morphism}=[fill=white, draw=black, shape=rectangle]
\tikzstyle{medium box}=[fill=white, draw=black, shape=rectangle, minimum width=0.8cm, minimum height=0.9cm]
\tikzstyle{large morphism}=[fill=white, draw=black, shape=rectangle, minimum width=1.7cm, minimum height=1cm]
\tikzstyle{bn}=[fill=black, draw=black, shape=circle, inner sep=1.5pt]
\tikzstyle{state}=[fill=white, draw=black, regular polygon, regular polygon sides=3, minimum width=0.8cm, shape border rotate=180, inner sep=0pt]
\tikzstyle{medium state}=[fill=white, draw=black, regular polygon, regular polygon sides=3, minimum width=1.3cm, inner sep=0pt, shape border rotate=180]
\tikzstyle{medium effect}=[fill=white, draw=black, regular polygon, regular polygon sides=3, minimum width=1.3cm, inner sep=0pt, shape border rotate=0]
\tikzstyle{large state}=[fill=white, draw=black, regular polygon, regular polygon sides=3, minimum width=2.2cm, shape border rotate=180, inner sep=0pt]
\tikzstyle{wn}=[fill=white, draw=black, shape=circle, inner sep=1.5pt]

\tikzstyle{arrow}=[->]
\tikzstyle{dashed box}=[-, dashed]
\tikzstyle{condition}=[draw=blue, dashed]

\pgfdeclarelayer{edgelayer}
\pgfdeclarelayer{nodelayer}
\pgfsetlayers{edgelayer,nodelayer,main}
\tikzstyle{none}=[]

\tikzstyle{observation}=[text=blue]
\tikzstyle{morphism}=[fill=white, draw=black, shape=rectangle]
\tikzstyle{medium box}=[fill=white, draw=black, shape=rectangle, minimum width=0.8cm, minimum height=0.9cm]
\tikzstyle{large morphism}=[fill=white, draw=black, shape=rectangle, minimum width=1.7cm, minimum height=1cm]
\tikzstyle{bn}=[fill=black, draw=black, shape=circle, inner sep=1.5pt]
\tikzstyle{state}=[fill=white, draw=black, regular polygon, regular polygon sides=3, minimum width=0.8cm, shape border rotate=180, inner sep=0pt]
\tikzstyle{medium state}=[fill=white, draw=black, regular polygon, regular polygon sides=3, minimum width=1.3cm, inner sep=0pt, shape border rotate=180]
\tikzstyle{medium effect}=[fill=white, draw=black, regular polygon, regular polygon sides=3, minimum width=1.3cm, inner sep=0pt, shape border rotate=0]
\tikzstyle{large state}=[fill=white, draw=black, regular polygon, regular polygon sides=3, minimum width=2.2cm, shape border rotate=180, inner sep=0pt]
\tikzstyle{wn}=[fill=white, draw=black, shape=circle, inner sep=1.5pt]

\tikzstyle{arrow}=[->]
\tikzstyle{dashed box}=[-, dashed]
\tikzstyle{condition}=[draw=blue, dashed]

\tikzset{baseline=(current  bounding  box.center)}

\newcommand{\C}{\mathbb C}
\newcommand{\R}{\mathbb R}
\newcommand{\A}{\mathbb A}

\newcommand{\N}{\mathbb N}
\renewcommand{\det}{\mathrm{det}}
\newcommand{\cond}{\mathbf{Cond}}
\newcommand{\obs}{\mathbf{Obs}}

\newcommand{\del}{\mathrm{del}}
\newcommand{\cpy}{\mathrm{copy}}
\newcommand{\col}{\mathrm{col}}

\newcommand{\term}[1]{\mathrm{#1}}

\newcommand{\id}{\mathrm{id}}
\newcommand{\Id}{\mathrm{Id}}
\newcommand{\swap}{\mathrm{swap}}
\newcommand{\letin}[2]{\term{let}\,{#1}\,={#2}\,\term{in}\,}
\newcommand{\s}{\,|\,}

\newcommand{\obsto}{\leadsto}
\newcommand{\red}{\vartriangleright}
\newcommand{\rv}{\mathrm{R}}
\newcommand{\unit}{\mathrm{I}}

\newcommand{\normal}{\mathrm{normal}}

\newcommand{\sem}[1]{[\![{#1}]\!]}
\newcommand{\gauss}{\mathbf{Gauss}}
\newcommand{\finstoch}{\mathbf{FinStoch}}
\newcommand{\borelstoch}{\mathbf{BorelStoch}}

\newcommand{\eq}{\mathrel{\scalebox{0.4}[1]{${=}$}{:}{\scalebox{0.4}[1]{${=}$}}}}
\newcommand{\eqo}{{:}\scalebox{0.5}[1]{${=}$}\,}
\newcommand{\eqs}{\eqo}

\usepackage{color}
\definecolor{deepblue}{rgb}{0,0,0.5}
\definecolor{deepred}{rgb}{0.6,0,0}
\definecolor{deepgreen}{rgb}{0,0.5,0}

\DeclareFixedFont{\ttb}{T1}{txtt}{bx}{n}{10} 
\DeclareFixedFont{\ttm}{T1}{txtt}{m}{n}{10}  

\newcommand*{\mlstinline}[1]{\text{\lstinline|#1|}}

\theoremstyle{plain}
\newtheorem{theorem}{Theorem}[section]
\newtheorem*{theorem*}{Theorem}
\newtheorem{corollary}[theorem]{Corollary}

\newtheorem{lemma}[theorem]{Lemma}
\newtheorem{proposition}[theorem]{Proposition}
\theoremstyle{definition}
\newtheorem{definition}[theorem]{Definition}

\newtheorem{example}[theorem]{Example}

\begin{document}

\title{Compositional Semantics for Probabilistic Programs with Exact Conditioning}

\author{\IEEEauthorblockN{Dario Stein}
\IEEEauthorblockA{University of Oxford, UK}
\and
\IEEEauthorblockN{Sam Staton}
\IEEEauthorblockA{University of Oxford, UK}}


%


\maketitle

\thispagestyle{plain}
\pagestyle{plain}

\begin{abstract}
We define a probabilistic programming language for Gaussian random variables with a first-class exact conditioning construct. We give operational, denotational and equational semantics for this language, establishing convenient properties like exchangeability of conditions. Conditioning on equality of continuous random variables is nontrivial, as the exact observation may have probability zero; this is \emph{Borel's paradox}. Using categorical formulations of conditional probability, we show that the good properties of our language are not particular to Gaussians, but can be derived from universal properties, thus generalizing to wider settings. We define the Cond construction, which internalizes conditioning as a morphism, providing general compositional semantics for probabilistic programming with exact conditioning.
\end{abstract}


%
\IEEEpeerreviewmaketitle

\section{Introduction}

Probabilistic programming is the paradigm of specifying complex statistical models as programs, and performing inference on them. There are two ways of expressing dependence on observed data, thus learning from them: \emph{soft constraints} and \emph{exact conditioning}. Languages like Stan \cite{stan} or WebPPL \cite{dippl} use a scoring construct for soft constraints, re-weighting program traces by observed likelihoods. Other frameworks like Hakaru \cite{narayanan2016probabilistic} or Infer.NET \cite{InferNET18} allow exact conditioning on data. In this paper we provide two semantic analyses of exact conditioning in a simple Gaussian language: a denotational semantics, and a equational axiomatic semantics, which we prove to coincide for closed programs. Our denotational semantics is based on a new and general construction on Markov categories, which, we argue, serve as a good framework for exact conditioning in probabilistic programs. 

\subsection{Case study: Reasoning about a Gaussian Programming Language with Exact Conditions}
\label{sec:introA}

Exact conditioning decouples the generative model from the data observations. Consider the following example for Gaussian process regression (a.k.a. \emph{kriging}): The prior \mlstinline{ys} is a $100$-dimensional multivariate normal (Gaussian) vector; we perform inference by fixing four observed datapoints via exact conditioning $(\eq)$.

\begin{lstlisting}
ys = gp_sample(n=100,kernel=rbf)
for (i,c) in observations:
  ys[i] =:= c
\end{lstlisting}

\begin{figure}[h]
	\centering
	\includegraphics[width=0.5\textwidth]{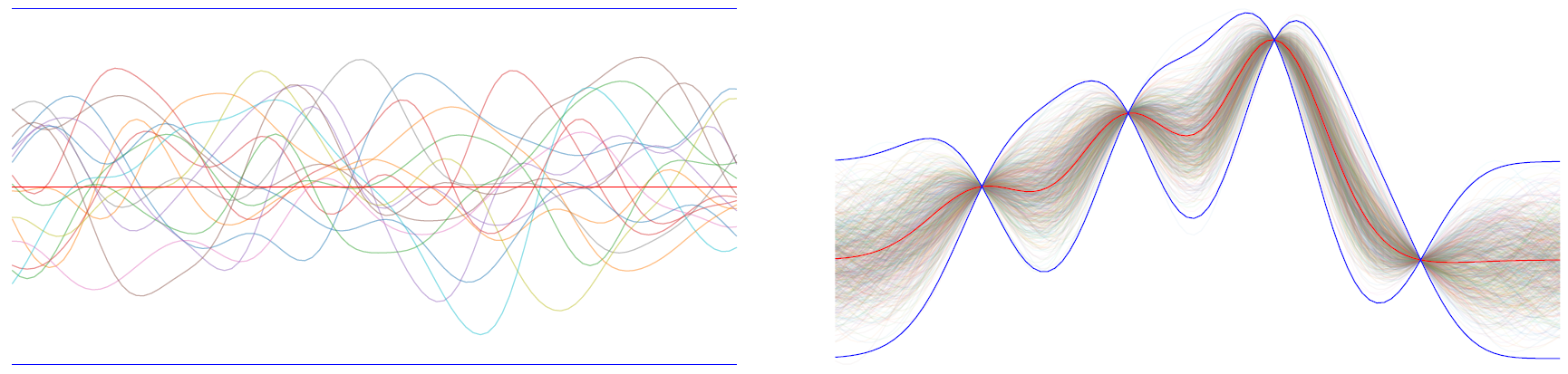}
	\caption{GP prior and posterior with 4 exact observations\label{fig:kriging}}
\end{figure}
The same program is difficult to express compositionally without exact conditioning.

No style of probabilistic modelling is immune to fallacies and paradoxes. Exact conditioning is indeed sensitive in this regard in general (\S\ref{sec:paradoxes}), and so it is important to show that where it is used, it is consistent in a compositional way. That is the contribution of this paper. 

The kriging example (Fig.~\ref{fig:kriging}) uses a smooth kernel, as is common, but to discuss the situation further we consider the following concrete variation with a Gaussian random walk. Suppose that the observation points are at $(0,20,40,60,80,100)$. 
\begin{lstlisting}
ys[0] = normal(0,1)
for i = 1 to 100:
  ys[i] = ys[i-1] + normal(0,1)
for j = 0 to 5:
  ys[20*j] =:= c[j]
\end{lstlisting}
To illustrate the power of compositional reasoning, we note that exact conditioning here is first-class, and as we will show, it is consistent to reorder programs as long as the dataflow is respected (Prop.~\ref{prop:good_properties}).
So this random walk program is equivalent to:
\begin{lstlisting}
ys[0] = normal(0,1)
ys[0] =:= c[0] 
for i = 1 to 100:
  ys[i] = ys[i-1] + normal(0,1)
  if i % 20 == 0: ys[i] =:= c[i % 20] 
\end{lstlisting}
We can now use a substitution law and initialization principle to simplify the program:
\begin{lstlisting}
ys[0] = c[0]
for i = 1 to 100:
  if i % 20 == 0: 
    ys[i] = c[i % 20] 
    (ys[i] - ys[i-1]) =:= normal(0,1) 
  else: ys[i] = ys[i-1] + normal(0,1)
\end{lstlisting}
The constraints are now all `soft', in that they relate an expression with a distribution, and so this last program could be run with a Monte Carlo simulation in Stan or WebPPL. Indeed, the soft-conditioning primitive \lstinline|observe| can be defined in terms of exact conditioning as 
\[ \mlstinline{observe(D,x)} \equiv \mlstinline{(let y = sample(D) in x =:= y)} \]
Our language is by no means restricted to kriging. For example, we can use similar techniques to implement and verify a simple K\'{a}lm\'{a}n filter (\ref{sec:impl}).  

In Section~\ref{sec:gaussian_language} we provide an operational semantics for this language, in which there are two key commands: drawing from a standard normal distribution ($\normal()$) and exact conditioning $(\eq)$. The operational semantics is defined in terms of configurations $(t,\psi)$ where $t$ is a program and $\psi$ is a state, which here is a Gaussian distribution. Each call to $\normal()$ introduces a new dimension into the state~$\psi$, and conditioning $(\eq)$ alters the state~$\psi$, using a canonical form of conditioning for Gaussian distributions (\S\ref{sec:recap_gauss}).

For the program in Figure~\ref{fig:kriging}, the operational semantics will first build up the prior distribution shown on the left in Figure~\ref{fig:kriging}, and then the second part of the program will condition to yield a distribution as shown on the right. But for the other programs above, the conditioning will be interleaved in the building of the model.

In stateful programming languages, composition of programs is often complicated and local transformations are difficult to reason about. But, as we now explain, we will show that for the Gaussian language, compositionality and local reasoning are straightforward. For example, as we have already  illustrated:
\begin{itemize}
\item Program lines can be reordered as long as dataflow is respected. That is, the following \emph{commutativity} equation \cite{staton:sfinite} remains valid for programs with conditioning
\begin{equation}
\begin{array}{l}
\letin x u \\
\letin y v \\
t
\end{array}
\equiv
\begin{array}{l}
\letin y v \\
\letin x u \\
t
\end{array}
\end{equation}
where $x$ not free in $v$ and $y$ not free in $u$.
\item We have a substitutivity property: if $t\eq u$ appears in a program then all other occurrences of $t$ can be replaced by $u$.
  \begin{equation}(t\eq u);v[^t\!/\!_x]\quad \equiv
    \quad (t\eq u);v[^u\!/\!_x]
    \label{eqn:substlong}
    \end{equation}
\item As a special base case, if we condition a normal variable on a constant, then the variable takes this value:
  \begin{equation}\letin x {\normal()} {(x\eq 0);t} \quad \equiv
    \quad t[^0\!/\!_x]
    \label{eqn:initlong}
    \end{equation}
\end{itemize}

\subsection{Denotational semantics, the Cond construction, and Markov categories}
In Section~\ref{sec:denotational} we show that this compositional reasoning is valid by using a denotational semantics. For a Gaussian language with\emph{out} conditioning, we can easily interpret terms as noisy affine functions, $x\mapsto Ax+c+\mathcal N(\Sigma)$. The exact conditioning requires a new construction for building a semantic model.
In fact this construction is not at all specific to Gaussian probability and works generally.

For this general construction, we start from a class of symmetric monoidal categories called \emph{Markov categories}~\cite{fritz}. These can be understood as the categorical counterpart of a broad class of probabilistic programming languages without conditioning (\S\ref{sec:internal_language}). For example, Gaussian probability forms a Markov category, but there are many other examples, including finite probability (e.g.~coin tosses) and full Borel probability.

Our conditioning construction starts from a Markov category~$\C$, regarded as a probabilistic programming language without conditioning. We build a new symmetric monoidal category $\cond(\C)$ which is conservative over $\C$ but which contains a conditioning construct. This construction builds on an analysis of conditional probabilities from the Markov category literature, which captures conditioning purely in terms of categorical structure: there is no explicit Radon-Nikod\'{y}m theorem, limits, reference measures or in fact any measure theory at all. The good properties of the Gaussian language generalize to this abstract setting, as they follow from universal properties alone.

The category $\cond(\C)$ has the same objects as $\C$, but a morphism is reminiscent of the decomposition of the program in Fig~\ref{fig:kriging}: a pair of a purely probabilistic morphism together with an observation. These morphisms compose by composing the generative parts and accumulating the observations (for a graphical representation, see Figure \ref{fig:obs_composition}). The morphisms are considered up-to a natural contextual equivalence.
We prove some general properties about $\cond(\C)$:\begin{enumerate}
\item Proposition~\ref{prop:faithful}: $\cond(\C)$ is consistent, in that no distinct unconditional distributions from $\C$ are equated in $\cond(\C)$.
\item Proposition~\ref{prop:good_properties}: $\cond(\C)$ allows programs to be reordered according to their dataflow graph, i.e.~it satisfies the interchange law of monoidal categories.
\end{enumerate}
Returning to the specific case study of Gaussian probability, we show that we have a canonical interpretation of the Gaussian language in $\cond(\gauss)$, which is fully abstract (Prop~\ref{prop:full_abstraction}). In consequence, the principles of reordering and consistency hold for the contextual equivalence induced by the operational semantics. 

\subsection{Equational axioms}
Our second semantic analysis (\S\ref{sec:equational_theory}) has a more syntactic and concrete flavor. We leave the generality of Markov categories and focus again on the Gaussian language. We present an equational theory for programs and use this to give normal forms for programs.

Our equational theory is surprisingly simple. The first two equations are
\begin{align*}
  &\letin x {\normal()} ()\quad = \quad ()\\[6pt]
  &\letin {x_1} {\normal(), \dots, x_n = \normal()}{U\vec x}  \\= \quad&  \letin {x_1} {\normal(), \dots, x_n =\normal()}{\vec x}
    \end{align*}
    The first equation is sometimes called discarding.
    In the second equation,~$U$ must be an orthogonal matrix, and we are using shorthand for multiplying a vector by a matrix.
    These two equations are enough to fully axiomatize the fragment of the Gaussian language without conditioning (Prop.~\ref{prop:gauss_complete}).
    
    (In Section~\ref{sec:equational_theory} we use a concise notation, writing the first axiom as $\nu x.\mathrm{r}[] = \mathrm{r}[]$. One instance of the second axiom with a permutation matrix for~$U$ is $\nu x.\nu y.\mathrm{r}[x,y] = \nu y.\nu x.\mathrm{r}[x,y]$, reminiscent of name generation in the $\pi$-calculus~\cite{piI} or $\nu$-calculus~\cite{pitts-stark-nu}.) 

The remaining axioms focus on conditioning. There are commutativity axioms for reordering parts of programs, as well as the two substitutivity axioms considered above, \eqref{eqn:substlong}, \eqref{eqn:initlong}. Finally there are two axioms for eliminating a condition that is tautologous $(a\eq a)$ or impossible $(0 \eq 1)$. 

Together, these axioms are consistent, which we can deduce by showing them to hold in the $\cond$ model.
To moreover illustrate the strength of the axioms, we show two normal form theorems by merely using the axioms. Here $\normal_n()$ describes the $n$-dimensional standard normal distribution.
\begin{itemize}
\item Proposition~\ref{prop:normalization}: any closed program is either derivably impossible $(0 \eq 1)$ or derivably equal to a condition-free program of the form ${A*\normal_n() + \vec c}$.
\item Theorem~\ref{prop:normalform_uniqueness}: any program of unit type (with no return value) is either derivably impossible $(0 \eq 1)$ or derivably equal to a soft constraint, i.e.~a program of the form $A*\vec x \eq  B*\normal_n() +\vec c$. We also give a uniqueness criterion on $A$, $B$ and $\vec c$.
  \end{itemize}

\subsection{Summary}
\begin{itemize}
\item We present a minimalist language with exact conditioning for Gaussian probability, with the purpose of studying the abstract properties of conditioning. Despite its simplicity, the language can express Gaussian processes or K\'{a}lm\'{a}n filters. 

\item In order to make the denotational semantics compositional, we introduce the Cond construction, which extends a Markov category $\C$ to a category $\cond(\C)$ in which conditioning internalizes as a morphism. The Gaussian language is recovered as the internal language of $\cond(\gauss)$.

\item We give three semantics for the language -- operational (\S\ref{sec:gaussian_language}), denotational (\S\ref{sec:denotational}) and axiomatic (\S\ref{sec:equational_theory}). We show that the denotational semantics is fully abstract (Proposition~\ref{prop:full_abstraction}) and that the axiomatic semantics is strong enough to derive normal forms (Theorem~\ref{prop:normalform_uniqueness}). This justifies properties like commutativity and substitutivity for the language. Thus probabilistic programming with exact conditioning can serve as a practical foundation for compositional statistical modelling. 

\end{itemize}
\section{A Language for Gaussian Probability}
\label{sec:gaussian_language}
We introduce a typed language (\S\ref{sec:types}), similar to the one discussed in Section~\ref{sec:introA},
and provide an operational semantics (\S\ref{sec:opsem}). 

\subsection{Recap of Gaussian Probability}
\label{sec:recap_gauss}
We briefly recall \emph{Gaussian probability}, by which we mean the treatment of multivariate Gaussian distributions and affine-linear maps (e.g. \cite{lauritzen}). A \emph{(multivariate) Gaussian distribution} is the law of a random vector $X \in \R^n$ of the form $X=AZ + \mu$ where $A \in \R^{n \times m}$, $\mu \in \R^n$ and the random vector $Z$ has components $Z_1, \ldots, Z_m \sim \mathcal N(0,1)$ which are independent and standard normally distributed with density function
\[ \varphi(x) = \frac{1}{\sqrt{2\pi}} e^{-\frac 1 2 x^2} \] The distribution of $X$ is fully characterized by its \emph{mean} $\mu$ and the positive semidefinite \emph{covariance matrix} $\Sigma$. Conversely, for any $\mu$ and positive semidefinite matrix $\Sigma$ there is a unique Gaussian distribution of that mean and covariance denoted $\mathcal N(\mu, \Sigma)$. The vector $X$ takes values precisely in the affine subspace $S=\mu + \col(\Sigma)$ where $\col(\Sigma)$ denotes the column space of $\Sigma$. We call $S$ the \emph{support} of the distribution.

This defines a small convenient fragment of probability theory: Affine transformations of Gaussians remain Gaussian. Furthermore, conditional distributions of Gaussians are again Gaussian. This is known as self-conjugacy. If $X \sim \mathcal N(\mu, \Sigma)$ with
\begin{equation*}
X = \begin{pmatrix} X_1 \\ X_2 \end{pmatrix}, \mu = \begin{pmatrix} \mu_1 \\ \mu_2 \end{pmatrix}, \Sigma = \begin{pmatrix} 
\Sigma_{11} & \Sigma_{12} \\ \Sigma_{21} & \Sigma_{22}
\end{pmatrix}
\end{equation*}
then the conditional distribution $X_1|(X_2 = a)$ of $X_1$ conditional on $X_2=a$ is $\mathcal N(\mu',\Sigma')$ where
\begin{equation}
\mu' = \mu_1 + \Sigma_{12}\Sigma_{22}^{+}(a-\mu_2), \Sigma' = \Sigma_{11} - \Sigma_{12}\Sigma_{22}^+\Sigma_{21}
\label{eq:conjugacy_formula}
\end{equation}
and $\Sigma_{22}^+$ denotes the Moore-Penrose pseudoinverse.

\begin{example}
\label{ex:og_gaussian_diagonal}
If $X,Y \sim \mathcal N(0,1)$ are independent and $Z=X-Y$, then 
\[ (X,Y)|(Z=0) \sim \mathcal N\left(\begin{pmatrix} 0 \\ 0 \end{pmatrix}, \begin{pmatrix} 0.5 & 0.5 \\ 0.5 & 0.5 \end{pmatrix} \right) \]
The posterior distribution is equivalent to the model
\[ X \sim \mathcal N(0,0.5), Y=X \]
\end{example}

\subsection{Types and terms of the Gaussian language}
\label{sec:types}
We now describe a language for Gaussian probability and conditioning. The core language resembles first-order OCaml with a construct $\normal()$ to sample from a standard Gaussian, and conditioning denoted as $(\eq)$. Types $\tau$ are generated from a basic type $\rv$ denoting \emph{real} or \emph{random variable}, pair types and unit type $I$. 
\[ \tau ::= \rv \s \unit \s \tau \ast \tau \]
Terms of the language are
\begin{align*}
e ::= x &\s e + e \s \alpha \cdot e \s \underline{\beta} \s (e,e) \s () \\
      &\s \letin x e e \s \letin {(x,y)} e e \\
      &\s \normal() \s e \eq e 
\end{align*}
where $\alpha,\beta$ range over real numbers. 

Typing judgements are
\[ \infer{\Gamma, x : \tau, \Gamma' \vdash x : \tau}{}  
\qquad
\infer{\Gamma \vdash () : \unit}{}
\qquad
\infer{\Gamma \vdash (s,t) : \sigma \ast \tau}{\Gamma \vdash s : \sigma \quad \Gamma \vdash t : \tau}  
\]
\[ \infer{\Gamma \vdash s + t : \rv}{\Gamma \vdash s : \rv \quad \Gamma \vdash t : \rv} 
\qquad
\infer{\Gamma \vdash \alpha \cdot t : \rv}{\Gamma \vdash t : \rv}
\qquad
\infer{\Gamma \vdash \underline{\beta} : \rv}{}
\]
\[ \infer{\Gamma \vdash \normal() : \rv}{} 
\qquad
\infer{\Gamma \vdash (s \eq t) : \unit}{\Gamma \vdash s : \rv \quad \Gamma \vdash t : \rv}
\]
\[ \infer{\Gamma \vdash \letin x s t : \tau}{\Gamma \vdash s : \sigma \quad \Gamma, x : \sigma \vdash t : \tau}
\]
\[ 
\infer{\Gamma \vdash \letin {(x,y)} s t : \tau}{\Gamma \vdash s : \sigma \ast \sigma' \quad \Gamma, x : \sigma, y : \sigma' \vdash t : \tau}
\]
We define standard syntactic sugar for sequencing $s;t$, identifying the type $\rv^n = \rv \ast (\rv \ast \ldots )$ with vectors and defining matrix-vector multiplication $A \cdot \vec x$. For $\sigma \in \R$ and $e : \rv$, we define $\normal(x,\sigma^2) \equiv x + \sigma \cdot \normal()$. More generally, for a covariance matrix $\Sigma$, we write $\normal(\vec x, \Sigma) = \vec x + A\cdot (\normal(), \ldots, \normal())$ where $A$ is any matrix such that $\Sigma = AA^T$. We can identify any context and type with $\rv^n$ for suitable $n$.

\subsection{Operational semantics}
\label{sec:opsem}

Our operational semantics is call-by-value. Calling $\normal()$ allocates a latent random variable, and a prior distribution over all latent variables is maintained. Calling $(\eq)$ updates this prior by symbolic inference according to the formula \eqref{eq:conjugacy_formula}.

\emph{Values} $v$ and redexes $\rho$ are defined as 
 \begin{align*}
v &::= x \s (v,v) \s v + v \s \alpha \cdot v \s \underline{\beta} \s () \\
\rho &::= \normal() \s v \eq v \s \letin x v e \s \letin {(x,y)} v e
 \end{align*}
A reduction context $C$ with hole $[-]$ is of the form
 \begin{align*}
 C ::= [-] &\s C + e \s v + C \s r \cdot C \s C \eq e \s v \eq C \\
 &\s \letin x C e \s \letin {(x,y)} C e
\end{align*}
Every term is either a value or decomposes uniquely as $C[\rho]$. We define a reduction relation for terms. During the execution, we will allocate latent variables $z_i$ which we assume distinct from all other variables in the program. A \emph{configuration} is either a pair $(e,\psi)$ where $z_1, \ldots, z_r \vdash e$ and $\psi$ is a Gaussian distribution on $\R^r$, or a failure configuration $\bot$. We first define reduction on redexes 
\begin{enumerate}
\item For $\normal()$, we add an independent latent variable to the prior
\[ (\normal(),\psi) \red (z_{\mathrm{r+1}}, \psi \otimes \mathcal N(0,1)) \]

\item To define conditioning, note that every value $z_1, \ldots, z_r \vdash v : R$ defines an affine function $\R^r \to \R$. In order to reduce $(v \eq w, \psi)$, we consider the joint distribution $X \sim \psi, Z = v(X)-w(X)$. If $0$ lies in the support of $Z$, we denote by $\psi|_{v=w}$ the outcome of conditioning $X$ on $Z=0$ as in \eqref{eq:conjugacy_formula}, and reduce
\[ (v \eq w, \psi) \red ((), \psi|_{v=w}) \]
Otherwise $(v \eq w, \psi) \red \bot$, indicating that the inference problem has no solution.
\item Let bindings are standard
\begin{align*} 
(\letin x v e, \psi) &\red (e[v/x], \psi) \\
(\letin {(x,y)} {(v,w)} e, \psi) &\red (e[v/x,w/y], \psi)
\end{align*}

\item Lastly, under reduction contexts, if $(\rho,\psi) \red (e,\psi')$ we define $(C[\rho],\psi) \red (C[e], \psi')$. If $(\rho,\psi) \red \bot$ then $(C[e], \psi) \red \bot$.
\end{enumerate}

\begin{proposition}
Every closed program $\vdash e : R^n$, together with the empty prior `$!$', deterministically reduces to either a configuration $(v,\psi)$ or $\bot$.
\end{proposition}
We consider the observable result of this execution either failure, or the pushforward distribution $v_*\psi$ on $\R^n$, as this distribution could be sampled from empirically.

\begin{example}
The program
\[ \letin {(x,y)} {(\normal(),\normal())} x \eq y; x+y \]
reduces to $((z_1,z_2), \psi)$ where 
\[ \psi = \mathcal N\left(\begin{pmatrix} 0 \\ 0 \end{pmatrix}, \begin{pmatrix} 0.5 & 0.5 \\ 0.5 & 0.5 \end{pmatrix} \right) \]
The observable outcome of the run is the pushforward distribution $(1\, 1)_*\psi = \mathcal N(0,2)$ on $\R$.
\end{example}

One goal of this paper is to study properties of this language compositionally, and abstractly, without relying on any specific properties of Gaussians. The crucial notion to investigate is contextual equivalence.

\begin{definition}\label{def:ctxequiv}
We say $\Gamma \vdash e_1, e_2 : \tau$ are \emph{contextually equivalent}, written $e_1 \approx e_2$, if for all closed contexts $K[-]$ and $i,j \in \{1,2\}$
\begin{enumerate}
\item when $(K[e_i], !) \red^* (v_i, \psi_i)$ then $(K[e_j], !) \red^* (v_j, \psi_j)$ and $(v_i)_*\psi_i = (v_j)_*\psi_j$
\item when $(K[e_i]) \red^* \bot$ then $(K[e_j],!) \red^* \bot$
\end{enumerate} 
\end{definition}

We study contextual equivalence by developing a denotational semantics for the Gaussian language (\S\ref{sec:denotational}), and proving it fully abstract (Prop.~\ref{prop:full_abstraction}). We furthermore show that these semantics can be axiomatized completely by a set of program equations (\S\ref{sec:equational_theory}). 

We also note nothing conceptually limits our language to only Gaussians. We are running with this example for concreteness, but any family of distributions which can be sampled and conditioned can be used. So we will take care to establish properties of the semantics in a general setting.
\section{Categorical Foundations of Conditioning}
We will now generalize away from Gaussian probability, recovering its convenient structure in the general categorical framework of Markov categories~(\S\ref{sec:markovcats}). We argue that this is a categorical counterpart of probabilistic programming without conditioning~(\S\ref{sec:internal_language}).

\begin{definition}[{\hspace{1sp}\cite[\S~6]{fritz}}]
The symmetric monoidal category $\gauss$ has objects $n \in \N$, which represent the affine space $\R^n$, and $m \otimes n = m+n$. Morphisms $m \to n$ are tuples $(A,b,\Sigma)$ where $A \in \R^{n \times m}, b \in \R^n$ and $\Sigma \in \R^{n \times n}$ is a positive semidefinite matrix. The tuple represents a stochastic map $f : \R^m \to \R^n$ that is affine-linear, perturbed with multivariate Gaussian noise of covariance $\Sigma$, informally written
\[ f(x) = Ax + b + \mathcal N(\Sigma) \]
Such morphisms compose sequentially and in parallel in the expected way, with noise accumulating independently
\begin{align*} (A,b,\Sigma) \circ (C,d,\Xi) &= (AC, Ad + b, A\Xi A^T + \Sigma) \\
(A,b,\Sigma) \otimes (C,d,\Xi) &= \left( 
\begin{pmatrix} A & 0 \\ 0 & C \end{pmatrix},
\begin{pmatrix} b \\ d \end{pmatrix},
\begin{pmatrix} \Sigma & 0 \\ 0 & \Xi \end{pmatrix} \right)
\end{align*}
\end{definition}
$\gauss$ furthermore has ability to introduce correlations and discard values by means of the affine maps $\cpy : \R^n \to \R^{n+n}, x \mapsto (x,x)$ and $\del : \R^n \to \R^0, x \mapsto ()$. This gives $\gauss$ the structure of a categorical model of probability, namely a Markov category. 

\subsection{Conditioning in Markov and CD categories}
\label{sec:markovcats}
\emph{Markov} and \emph{CD categories} are a formalism that is increasingly widely used (e.g. \cite{representable_markov, infinite_products}). We review their graphical language, and theory of conditioning.

\begin{definition}[{\hspace{1sp}\cite{cho-jacobs}}]
A \emph{copy-delete (CD) category} is a symmetric monoidal category $(\C, \otimes, I)$ in which every object $X$ is equipped with the structure of a commutative comonoid $\cpy_X : X \to X \otimes X$, $\del_X : X \to I$ which is compatible with the monoidal structure.
\end{definition}

In CD categories, morphisms $f : X \to Y$ need not be \emph{discardable}, i.e. satisfy $\del_Y \circ f = \del_X$. If they are, we obtain a Markov category. 
\begin{definition}[{\hspace{1sp}\cite{fritz}}]
A \emph{Markov category} is a CD category in which every morphism is discardable, i.e. $\del$ is natural. Equivalently, the unit $I$ is terminal. 
\end{definition}

Beyond $\gauss$, further examples of Markov categories are the category $\finstoch$ of finite sets and stochastic matrices, and the category $\borelstoch$ of Markov kernels between standard Borel spaces. CD categories generalize \emph{unnormalized} measure kernels.

The interchange law of $\otimes$ encodes exchangeability (Fubini's theorem) while the discardability condition signifies that probability measures are normalized to total mass $1$. We introduce the following terminology: States $\mu : I \to X$ are also called \emph{distributions}, and if $f : A \to X \otimes Y$, we denote its \emph{marginals} by $f_X : A \to X, f_Y : A \to Y$. Copying and discarding allows us to write tupling $\langle f, g\rangle$ and projection $\pi_X$, however note that the monoidal structure is only semicartesian, i.e. $f \neq \langle f_X, f_Y \rangle$ in general. We use string diagram notation for symmetric monoidal categories, and denote the comonoid structure as

\[ \begin{tikzpicture}[scale=0.6]
	\begin{pgfonlayer}{nodelayer}
		\node [style=bn] (0) at (3, 1) {};
		\node [style=bn] (1) at (-1.75, 0) {};
		\node [style=none] (2) at (3, -1) {};
		\node [style=none] (3) at (-1.75, -1) {};
		\node [style=none] (4) at (-2.75, 1) {};
		\node [style=none] (5) at (-0.75, 1) {};
		\node [style=none] (10) at (2, 0) {$=$};
		\node [style=none] (11) at (-3.75, 0) {$=$};
		\node [style=none] (12) at (1, 0) {$\del_X$};
		\node [style=none] (13) at (-5.25, 0) {$\cpy_X$};
	\end{pgfonlayer}
	\begin{pgfonlayer}{edgelayer}
		\draw [bend right=45, looseness=1.25] (4.center) to (1);
		\draw (1) to (3.center);
		\draw [bend right=45] (1) to (5.center);
		\draw (2.center) to (0);
	\end{pgfonlayer}
\end{tikzpicture} \]

\begin{definition}[{\hspace{1sp}\cite[10.1]{fritz}}]
A morphism $f : X \to Y$ is called \emph{deterministic} if it commutes with copying, that is
\[ \cpy_Y \circ f = (f \otimes f) \circ \cpy_X \]
In a Markov category, the wide subcategory $\C_\det$ of deterministic maps \emph{is} cartesian, i.e. $\otimes$ is a product. 
\end{definition}
A morphism $(A,b,\Sigma)$ in $\gauss$ is deterministic iff $\Sigma=0$. The deterministic subcategory $\A = \gauss_\det$ consists of the spaces $\R^n$ and affine maps $x \mapsto Ax+b$ between them. \\

We recall the theory of conditioning for Markov categories. 

\begin{definition}[{\hspace{1sp}\cite[11.1,11.5]{fritz}}]
\label{def:conditional}
A \emph{conditional distribution} for $\psi : I \to X \otimes Y$ is a morphism $\psi|_X : X \to Y$ such that 
\begin{equation}
\label{eq:def_cond}
\begin{tikzpicture}[scale=0.5]
	\begin{pgfonlayer}{nodelayer}
		\node [style=medium state] (0) at (-2.5, -1) {$\psi$};
		\node [style=none] (1) at (0, 0) {$=$};
		\node [style=medium state] (2) at (3.25, -2) {$\psi$};
		\node [style=morphism] (3) at (4, 1) {$\psi_{|X}$};
		\node [style=none] (4) at (1.5, 0.75) {};
		\node [style=none] (5) at (3.75, -1.5) {};
		\node [style=none] (6) at (-1.75, -0.5) {};
		\node [style=none] (7) at (2.75, -1.5) {};
		\node [style=none] (8) at (-3.25, -0.5) {};
		\node [style=none] (9) at (-3.25, 1.5) {};
		\node [style=none] (10) at (-1.75, 1.5) {};
		\node [style=bn] (11) at (3.75, -0.75) {};
		\node [style=bn] (12) at (2.75, -0.5) {};
		\node [style=none] (13) at (4, 2.25) {};
		\node [style=none] (14) at (1.5, 2.25) {};
		\node [style=none] (15) at (-3.25, 2) {$X$};
		\node [style=none] (16) at (-1.75, 2) {$Y$};
		\node [style=none] (17) at (1.5, 2.75) {$X$};
		\node [style=none] (18) at (4, 2.75) {$Y$};
	\end{pgfonlayer}
	\begin{pgfonlayer}{edgelayer}
		\draw (9.center) to (8.center);
		\draw (10.center) to (6.center);
		\draw (11) to (5.center);
		\draw (12) to (7.center);
		\draw [bend left=45] (12) to (4.center);
		\draw [in=-90, out=0, looseness=0.75] (12) to (3);
		\draw (3) to (13.center);
		\draw (14.center) to (4.center);
	\end{pgfonlayer}
\end{tikzpicture}
\end{equation}
A \emph{(parameterized) conditional} for $f : A \to X \otimes Y$ is a morphism $f|_X : X \otimes A \to Y$ such that
\begin{equation*}
\begin{tikzpicture}[scale=0.3]
	\begin{pgfonlayer}{nodelayer}
		\node [style=none] (1) at (5, 2) {};
		\node [style=none] (2) at (4.5, 5.5) {};
		\node [style=none] (3) at (4.5, 4) {};
		\node [style=none] (7) at (2, 6) {$X$};
		\node [style=medium box] (14) at (4.5, 3) {$f_{|X}$};
		\node [style=none] (20) at (3, 0.5) {};
		\node [style=none] (21) at (4, -1) {};
		\node [style=none] (22) at (3.5, -3) {};
		\node [style=none] (23) at (4, 2) {};
		\node [style=bn] (24) at (4, 0) {};
		\node [style=none] (26) at (-1.75, 0.25) {};
		\node [style=none] (29) at (-3.25, 0.25) {};
		\node [style=bn] (33) at (5, -4.5) {};
		\node [style=none] (34) at (6.5, -3.5) {};
		\node [style=none] (35) at (5, -5.5) {};
		\node [style=none] (37) at (0, 0) {$=$};
		\node [style=none] (38) at (-3.25, 3.75) {};
		\node [style=bn] (39) at (3, 0.5) {};
		\node [style=none] (40) at (3, -1) {};
		\node [style=none] (41) at (2, 5.5) {};
		\node [style=none] (42) at (-2.5, -4.25) {};
		\node [style=none] (43) at (-1.75, 3.75) {};
		\node [style=medium box] (48) at (-2.5, -0.25) {$f$};
		\node [style=medium box] (49) at (3.5, -2) {$f$};
		\node [style=none] (50) at (-2.5, -4.75) {$A$};
		\node [style=none] (51) at (-3.25, 4.25) {$X$};
		\node [style=none] (52) at (-1.75, 4.25) {$Y$};
		\node [style=none] (53) at (4.5, 6) {$Y$};
		\node [style=none] (54) at (5, -6) {$A$};
	\end{pgfonlayer}
	\begin{pgfonlayer}{edgelayer}
		\draw [style=none] (3.center) to (2.center);
		\draw [style=none, in=0, out=-90] (23.center) to (20.center);
		\draw [style=none] (24) to (21.center);
		\draw [style=none, in=-90, out=180, looseness=0.75] (33) to (22.center);
		\draw [style=none, in=-90, out=0] (33) to (34.center);
		\draw [style=none, in=-90, out=90, looseness=0.50] (34.center) to (1.center);
		\draw [style=none] (35.center) to (33);
		\draw [style=none, in=180, out=-90, looseness=0.50] (41.center) to (20.center);
		\draw [style=none] (20.center) to (40.center);
		\draw [style=none] (29.center) to (38.center);
		\draw [style=none] (26.center) to (43.center);
		\draw (42.center) to (48);
	\end{pgfonlayer}
\end{tikzpicture}
\end{equation*}
\end{definition}

Parameterized conditionals can be specialized to conditional distributions in the following way
\begin{proposition}
\label{prop:specialization}
If $f : A \to X \otimes Y$ has conditional $f|_X : X \otimes A \to Y$ and $a : I \to X$ is a deterministic state, then $f|_X(\id_X \otimes a)$ is a conditional distribution for $fa$.
\end{proposition}

All of our examples $\finstoch$, $\borelstoch$ and $\gauss$ have conditionals\cite{fritz, bogachev}. For $\gauss$, this captures the self-conjugacy of Gaussians \cite{jacobs:conjugatepriors}. An explicit formula generalizing \eqref{eq:conjugacy_formula} is given in \cite{fritz}, but we shall only require the existence of conditionals and work with their universal property. 

\begin{definition}[{\hspace{1sp}\cite[13.1]{fritz}}]
Let $\mu : I \to X$ be a distribution. Parallel morphisms $f,g : X \to Y$ are called \emph{$\mu$-almost surely equal}, written $f=_\mu g$, if $ \langle \id_X, f \rangle \mu = \langle \id_X, g \rangle \mu$.
\end{definition}

Conditional distributions for a given distribution $\mu : I \to X \otimes Y$ are generally not unique. However, it follows from definition that they are $\mu_X$-almost surely equal. In order to uniquely evaluate conditionals at a point, we need to descend from the global universal property to individual inputs. This is achieved by the absolute continuity relation.

\begin{definition}[{\hspace{1sp}\cite[2.8]{representable_markov}}]
Let $\mu, \nu : I \to X$ be two distributions. We write $\mu \ll \nu$ if for all $f, g : X \to Y$, $f=_\nu g$ implies $f=_\mu g$.
\end{definition}

\begin{lemma}
\label{lemma:llunique}
If $f, g : X \to Y$ are $\mu$-almost surely equal and $x : I \to X$ satisfies $x \ll \mu$ then $fx=gx$.
\end{lemma}

\begin{proposition}
\label{prop:gaussian_supports}
For a distribution $\mu = \mathcal N(b,\Sigma) : 0 \to m$ in $\gauss$, let $S = b+\mathrm{col}(\Sigma)$ be its support as in \S\ref{sec:recap_gauss}. Then
\begin{itemize}
\item If $f, g : m \to n$ are morphisms, then $f =_\mu g$ iff $fx = gx$ for all $x \in S$, seen as deterministic states $x : 0 \to m$.
\item If $\nu : 0 \to m$ then $\mu \ll \nu$ iff the support of $\mu$ is contained in the support of $\nu$
\item In particular for $x : 0 \to m$ deterministic, $x \ll \mu$ iff $x \in S$.
\end{itemize}
\end{proposition}
There is a general notion of support in Markov categories defined in \cite{fritz} which agrees with $S$, but we will formulate our results in terms of the more flexible notion $\ll$.
\begin{proof}
See appendix, where we also characterize $\ll$ for $\finstoch$ and $\borelstoch$.
\end{proof}

We give an example of how to use the categorical conditioning machinery in practice.
\begin{example}
\label{ex:gaussian_diagonal}
The statistical model from Example \ref{ex:og_gaussian_diagonal}
\begin{align*}
X \sim \mathcal N(0,1);\,
Y \sim \mathcal N(0,1);\,
Z = X - Y
\end{align*}
corresponds to the distribution $\mu : 0 \to 3$ with covariance matrix
\[ \Sigma = \begin{pmatrix} 1 & 0 & 1 \\ 0 & 1 & 1 \\ 1 & 1 & 2 \end{pmatrix} \]
A conditional with respect to $Z$ is
\[ \mu|_Z(z) = \begin{pmatrix}0.5 \\ 0.5\end{pmatrix}z + \mathcal N\begin{pmatrix}
0.5 & 0.5 \\ 0.5 & 0.5
\end{pmatrix} \]
which can be verified by calculating \eqref{eq:def_cond}. We wish to condition on $Z=0$. The marginal $\mu_Z = \mathcal N(2)$ is supported on all of $\R$, hence $0 \ll \mu_Z$ and by Lemma \ref{lemma:llunique} the composite 
\[ \mu|_{Z}(0) = \mathcal N\begin{pmatrix}
0.5 & 0.5 \\ 0.5 & 0.5
\end{pmatrix} \]
is uniquely defined and represents the posterior distribution over $(X,Y)$.
\end{example}
\subsection{Internal language of Markov categories}
\label{sec:internal_language}

There is a strong correspondence between first-order probabilistic programming languages and the categorical models of probability, via their internal languages. The internal language of a CD category~$\C$ has types 

\[ \tau ::= X \s I \s \tau \ast \tau \]

where $X$ ranges over objects of $\C$. Any type~$\tau$ can be regarded as an object $\sem \tau$ of $\C$, 
via $\sem{X}=X$, $\sem{I}=I$, and $\sem{\tau_1 \ast\tau_2}=
\sem{\tau_1}\otimes \sem{\tau_2}$. 
The terms of the internal language are like the language of Section~\ref{sec:gaussian_language}, built from $\letin x t u$, free variables and pairing, but 
instead of Gaussian-specific constructs like $\normal()$, $+$, and $\eq$, we have terms for any morphisms in~$\C$:
\[
    \infer[(f:\sem{\tau_1}{\otimes} \dots{\otimes} \sem{\tau_n}\to \sem{\tau'}\text{ in $\C$})]{\Gamma\vdash f(t_1\dots t_n) : \tau'}
    {\Gamma\vdash t_1:\tau_1\ \dots\ \Gamma\vdash t_n:\tau_n}
\]
Taking $\C=\gauss$ we recover the conditioning-free fragment of the language of Section~\ref{sec:gaussian_language} (\ref{ex:gauss_interpretation}), but the syntax makes sense for any CD or Markov category. A core result of this work is that the full language can be recovered as well for a CD category $\C=\cond(\gauss)$ (\S\ref{sec:cond}).

A typing context $\Gamma={(x_1:\tau_1\dots x_n:\tau_n)}$ is interpreted as $\sem{\Gamma}=\sem{\tau_1}\otimes \dots\otimes \sem{\tau_n}$. A term in context $\Gamma\vdash t :\tau$ 
is interpreted as a morphism $\sem\Gamma\to \sem \tau$, 
defined by induction on the structure of typing derivations. 
This is similar to the interpretation of a dual linear $\lambda$-calculus in a monoidal category \cite[\S3.1,\S4]{barberdill}, although because every type supports copying and discarding we do not need to distinguish between linear and non-linear variables. For example, 
\begin{align*}&\sem{\letin x t u}=
\sem\Gamma {\xrightarrow{\cpy}}  \sem\Gamma\otimes \sem\Gamma
{\xrightarrow{\sem t \otimes \id}} \sem A \otimes \sem\Gamma 
{\xrightarrow {\sem u}} \sem B
\\&\sem{\Gamma,x:\tau,\Gamma'\vdash x:\tau}=
\sem\Gamma\otimes \sem\tau\otimes\sem{\Gamma'}
\xrightarrow{\del\otimes\id_{\sem{\tau}}\otimes \del}\sem\tau
\end{align*}
The interpretation always satisfies the following identity, associativity and commutativity equations:
\begin{align}
&\sem{\letin y {(\letin x t u)} v} = \sem{\letin x t \letin y u v} \notag \\
&\sem{\letin x t x} = \sem t\quad\quad \sem{\letin x x u} = \sem u \label{eqn:commutativity} \\
&\sem{\letin x t {\letin y u v}} = \sem{\letin y u {\letin x t v}}\notag 
\end{align}
where $x$ not free in $u$ and $y$ not free in $t$. There are also standard equations for tensors~\cite[\S3.1]{statonlevy}, which always hold. 

We can always substitute terms for free variables: if we have $\Gamma,x:A\vdash t : B$ and $\Gamma\vdash u : A$ 
then $\Gamma\vdash t[^u\!/\!_x]:B$. 
In any CD category we have
\[\sem{\letin x t u}=\sem{u[^t\!/\!_x]}\text{ if $x$ occurs exactly once in $u$.}\]
In a Markov category, moreover, every term is discardable: 
\[\sem{\letin x t u}=\sem{u[^t\!/\!_x]}\text{ if $x$ occurs at most once in $u$.}\]
(It is common to also define a term to be \emph{copyable} if a version of the substitution condition holds when $x$ occurs \emph{at least once} (e.g.~\cite{fuhrmann,kammarplotkin}), but we will not need that in what follows.)

\begin{example}
\label{ex:gauss_interpretation}
The fragment of the Gaussian language without conditioning ($\eq$) is a subset of the internal language of the category $\gauss$.
That is to say, there is a canonical denotational semantics of the Gaussian language where we interpret types and contexts as objects of $\gauss$, e.g. $\sem{\rv}=1$ and $\sem{(x:\rv,y:\rv \otimes \rv)} = 3$. Terms $\Gamma\vdash t:A$ are interpreted as stochastic maps $Ax+b+ \mathcal N(\Sigma)$. This is all automatic once we recognize that addition $(+) : 2 \to 1$, scaling $\alpha \cdot (-) : 1 \to 1$, constants $\underline{\beta} : 0 \to 1$ and sampling $\mathcal N(1) : 0 \to 1$ are morphisms in $\gauss$.
\end{example}

\begin{example} In Section~\ref{sec:cond}, we will show that the full Gaussian language with conditioning ($\eq$) is the internal language of a CD category. 
 The fact that commutativity~\eqref{eqn:commutativity} holds is non-trivial. It cannot reasonably be the internal language of a Markov category, because conditions $(\eq)$ cannot be discardable. For example there is no non-trivial morphism $(\eq) : 2 \to 0$ in $\gauss$. \end{example}
\section{Cond -- Compositional Conditioning}
\label{sec:cond}

Let $\C$ be a Markov category with conditionals (\S\ref{sec:markovcats}). For simplicity of notation, we assume $\C$ to be \emph{strictly monoidal}.

We construct a new category $\cond(\C)$ by adding to this category the ability to condition on fixed observations. By \emph{observation} we mean a deterministic state $o : I \to X$, and we seek to add for each of those a conditioning effect $(\eqo o) : X \to I$.

Our constructions proceed in two stages. We first (\S\ref{sec:obscat}) form a category $\obs(\C)$ on the same objects as $\C$ where $(\eqo o)$ is added purely formally. A morphism $X \obsto Y$ in $\obs(C)$ represents an intensional open program of the form
\begin{equation}
\label{eq:obs_normalform}
x : X \vdash \letin {(y,k) : Y \otimes K} {f(x)} {(k \eqo o); y}
\end{equation}
We think of $K$ as an additional hidden output wire, to which we attach the observation $o$. Such programs compose the obvious way, by aggregating observations (see Fig.~\ref{fig:obs_composition}).

In the second stage (\S\ref{sec:condcat}) -- this is the core of the paper -- we relate such open programs to the conditionals present in~$\C$, that is we quotient by contextual equivalence. The resulting quotient is called $\cond(\C)$. Under sufficient assumptions, this will have the good properties of a CD category. 

\subsection{Step 1 (Obs): Adding conditioning}\label{sec:obscat}
\begin{definition}
\label{def:obs}
The following data define a symmetric premonoidal category $\obs(\C)$:
\begin{itemize}
	\item the object part of $\obs(\C)$ is the same as $\C$ 
	\item morphisms $X \obsto Y$ are tuples $(K,f,o)$ where $K \in \mathrm{ob}(\C)$, $f \in \C(X,Y\otimes K)$ and $o \in \C_{\mathrm{det}}(I,K)$ 
	\item The identity on $X$ is $\Id_X = (I,\id_X,!)$ where $!=\id_I$.
	\item Composition is defined by
	\[ (K',f',o') \bullet (K,f,o) = (K' \otimes K, (f' \otimes \id_K)f, o' \otimes o). \]
	\item if $(K,f,o) : X \obsto Y$ and $(K',f',o') : X' \obsto Y'$, their tensor product is defined as
	\[ (K' \otimes K, (\id_{Y'} \otimes \swap_{K',Y} \otimes \id_K)(f' \otimes f), o' \otimes o) \]
	\item There is an identity-on-objects functor $J : \C \to \obs(\C)$ that sends $f : X \to Y$ to $(I,f,!) : X \obsto Y$. This functor is strict premonoidal and its image central
	\item $\obs(\C)$ inherits symmetry and comonoid structure
\end{itemize}
\end{definition}

A premonoidal category (due to \cite{power:premonoidal}) is like a monoidal category where the interchange law need not hold. This is the case because $\obs(\C)$ does not yet identify observations arriving in different order. This will be remedied automatically later when passing to the quotient $\cond(\C)$.

Composition and tensor can be depicted graphically as in Figure \ref{fig:obs_composition}, where dashed wires indicate condition wires $K$ and their attached observations $o$. For an observation $o : I \to K$, the conditioning effect $(\eqo o) : K \obsto I$ is given by $(I,\id_K,o)$.

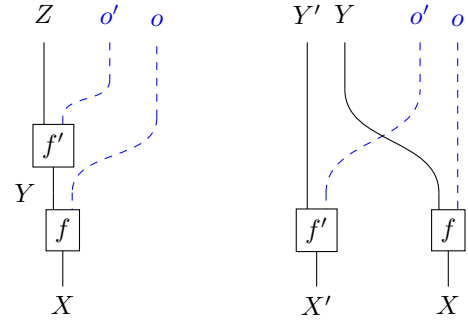
\begin{figure}[h]
	\centering
	\begin{tikzpicture}[scale=0.5]
	\begin{pgfonlayer}{nodelayer}
		\node [style=morphism] (19) at (-2.75, 0) {$f$};
		\node [style=none] (20) at (-2.5, 1) {};
		\node [style=none] (21) at (-3, 1) {};
		\node [style=none] (22) at (-3, 2) {};
		\node [style=none] (23) at (-0.25, 3) {};
		\node [style=none] (24) at (-2.75, -1.5) {};
		\node [style=none] (25) at (-2.75, -0.5) {};
		\node [style=morphism] (26) at (-3, 2.25) {$f'$};
		\node [style=none] (28) at (-0.25, 5) {};
		\node [style=none] (29) at (-2.75, 3) {};
		\node [style=none] (30) at (-3.25, 3) {};
		\node [style=none] (31) at (-3.25, 5) {};
		\node [style=none] (32) at (-1.5, 4) {};
		\node [style=none] (35) at (-1.5, 5) {};
		\node [style=none] (36) at (-3.25, 5.75) {$Z$};
		\node [style=observation] (37) at (-1.5, 5.75) {$o'$};
		\node [style=observation] (38) at (-0.25, 5.6) {$o$};
		\node [style=none] (39) at (-3.75, 1) {$Y$};
		\node [style=none] (40) at (-2.75, -2) {$X$};
		\node [style=none] (41) at (-2.5, 0.25) {};
		\node [style=none] (42) at (-3, 0.25) {};
		\node [style=none] (43) at (-2.75, 2.5) {};
		\node [style=none] (44) at (-3.25, 2.5) {};
		\node [style=morphism] (45) at (4, 0) {$f'$};
		\node [style=none] (46) at (4.25, 1) {};
		\node [style=none] (47) at (3.75, 5) {};
		\node [style=none] (48) at (4, -1.5) {};
		\node [style=none] (49) at (4, -0.5) {};
		\node [style=none] (50) at (4, -2) {$X'$};
		\node [style=none] (51) at (4.25, 0.25) {};
		\node [style=none] (52) at (3.75, 0.25) {};
		\node [style=morphism] (53) at (7.5, 0) {$f$};
		\node [style=none] (54) at (7.75, 5) {};
		\node [style=none] (55) at (7.25, 1) {};
		\node [style=none] (56) at (7.5, -1.5) {};
		\node [style=none] (57) at (7.5, -0.5) {};
		\node [style=none] (58) at (7.5, -2) {$X$};
		\node [style=none] (59) at (7.75, 0.25) {};
		\node [style=none] (60) at (7.25, 0.25) {};
		\node [style=none] (61) at (4.25, 1) {};
		\node [style=none] (62) at (7.25, 1) {};
		\node [style=none] (63) at (6.75, 3.75) {};
		\node [style=none] (64) at (4.75, 3.75) {};
		\node [style=observation] (65) at (6.75, 5.75) {$o'$};
		\node [style=observation] (66) at (7.75, 5.6) {$o$};
		\node [style=none] (68) at (4.75, 5) {};
		\node [style=none] (69) at (6.75, 5) {};
		\node [style=none] (70) at (3.75, 5.75) {$Y'$};
		\node [style=none] (71) at (4.75, 5.75) {$Y$};
	\end{pgfonlayer}
	\begin{pgfonlayer}{edgelayer}
		\draw (22.center) to (21.center);
		\draw [style=condition, in=90, out=-90] (23.center) to (20.center);
		\draw (25.center) to (24.center);
		\draw (31.center) to (30.center);
		\draw [style=condition, in=90, out=-90] (32.center) to (29.center);
		\draw [style=condition] (23.center) to (28.center);
		\draw (21.center) to (42.center);
		\draw [style=condition] (20.center) to (41.center);
		\draw (30.center) to (44.center);
		\draw [style=condition] (29.center) to (43.center);
		\draw [style=condition] (35.center) to (32.center);
		\draw (49.center) to (48.center);
		\draw (47.center) to (52.center);
		\draw [style=condition] (46.center) to (51.center);
		\draw (57.center) to (56.center);
		\draw (55.center) to (60.center);
		\draw [style=condition] (54.center) to (59.center);
		\draw [style=condition, in=-90, out=90] (61.center) to (63.center);
		\draw [in=90, out=-90] (64.center) to (62.center);
		\draw (68.center) to (64.center);
		\draw [style=condition] (69.center) to (63.center);
	\end{pgfonlayer}
\end{tikzpicture}
	\caption{Composition and tensoring of morphisms in $\obs$}
	\label{fig:obs_composition}
\end{figure}

\subsection{Step 2 (Cond): Equivalence of open programs}
\label{sec:condcat}

We now quotient $\obs$-morphisms, tying them to the conditionals which can be computed in $\C$. We know how to compute conditionals for closed programs. Given a state $(K,\psi,o) : I \obsto m$, we follow the procedure of Example \ref{ex:gaussian_diagonal}: If $o \not \ll \psi_K$, the observation does not lie in the support of the model and conditioning fails. If not, we form the conditional $\psi|_K$ in $\C$ and obtain a well-defined posterior $\mu|_K \circ o$. 

This notion defines an equivalence relation on states $I \leadsto n$ in $\cond(\C)$. We will then extend this notion to a congruence on arbitrary morphisms $X \leadsto Y$ by a general categorical construction. 

\begin{definition}
\label{def:condeq}
Given two states $I \leadsto X$ we define $(K,\psi,o) \sim (K',\psi',o')$ if either
\begin{enumerate}
\item $o \ll \psi_K$ and $o' \ll \psi'_{K'}$ and $\psi|_K(o)= \psi'|_{K'}(o')$.
\item $o \not \ll \psi_K$ and $o' \not \ll \psi'_{K'}$
\end{enumerate}
That is, both conditioning problems either fail, or both succeed with equal posterior. 
\end{definition}
Figure~\ref{fig:obs} formulates Example \ref{ex:gaussian_diagonal} in $\obs(\gauss)$:

\begin{figure}[h]
	\vspace{-12pt}
	\centering
	\begin{tikzpicture}[scale=0.4]
	\begin{pgfonlayer}{nodelayer}
		\node [style=none] (7) at (-5.25, -1.5) {};
		\node [style=none] (9) at (-2.75, -1.5) {};
		\node [style=none] (13) at (-6.25, 2) {};
		\node [style=bn] (14) at (-5.25, 0) {};
		\node [style=bn] (21) at (-2.75, 0) {};
		\node [style=none] (22) at (-1.75, 2) {};
		\node [style=none] (23) at (-2.25, 2) {};
		\node [style=none] (24) at (-5.75, 2) {};
		\node [style=morphism] (25) at (-2, 2.5) {$-$};
		\node [style=none] (28) at (-2, 4) {};
		\node [style=none] (29) at (-6.25, 3.75) {};
		\node [style=none] (30) at (-5.75, 3.75) {};
		\node [style=observation] (31) at (-2, 4.75) {$0$};
		\node [style=none] (33) at (3, -1.5) {};
		\node [style=none] (47) at (0, 0) {};
		\node [style=none] (48) at (0, 0) {$\sim$};
		\node [style=none] (49) at (-5.25, -2) {$\mathcal N(1)$};
		\node [style=none] (50) at (-2.75, -2) {$\mathcal N(1)$};
		\node [style=none] (51) at (3, -2) {$\mathcal N(0.5)$};
		\node [style=bn] (52) at (3, 0) {};
		\node [style=none] (53) at (1.75, 2.25) {};
		\node [style=none] (54) at (4.25, 2.25) {};
	\end{pgfonlayer}
	\begin{pgfonlayer}{edgelayer}
		\draw [bend left] (14) to (13.center);
		\draw [bend left] (22.center) to (21);
		\draw (7.center) to (14);
		\draw [bend left=45] (21) to (24.center);
		\draw [bend right=45] (14) to (23.center);
		\draw (21) to (9.center);
		\draw [style=condition] (28.center) to (25);
		\draw (13.center) to (29.center);
		\draw (24.center) to (30.center);
		\draw (33.center) to (52);
		\draw [bend left] (52) to (53.center);
		\draw [bend right=45, looseness=0.75] (52) to (54.center);
	\end{pgfonlayer}
\end{tikzpicture}
	\caption{Example \ref{ex:gaussian_diagonal} describes related states $0 \obsto 2$}
	\label{fig:obs}
\end{figure}

\begin{definition}
\label{def:observational_quotient}
Let $\mathbb X$ be a symmetric premonoidal category. An equivalence relation $\sim$ on states $\mathbb X(I,-)$ is called \emph{functorial} if $\psi \sim \psi'$ implies $f\psi \sim f\psi'$. We can extend such a relation to a congruence $\approx$ on all morphisms $X \to Y$ via
\[ f \approx g \Leftrightarrow \forall A,\psi : I \to A \otimes X, (\id_A \otimes f)\psi \sim (\id_A \otimes g)\psi. \]
The quotient category $\mathbb X/\!\!\approx$ is symmetric premonoidal. 
\end{definition}

We show now that under good assumptions, the quotient by conditioning \ref{def:condeq} on $\mathbb X = \obs(\C)$ is functorial, and induces a quotient category $\cond(\C)$. The technical condition is that supports interact well with dataflow
\begin{definition}
A Markov category $\C$ has \emph{precise supports} if the following are equivalent for all deterministic $x : I \to X$, $y : I \to Y$, and arbitrary $f : X \to Y$ and $\mu : I \to X$.
\begin{enumerate}
\item $x \otimes y \ll \langle \id_X, f\rangle \mu$
\item $x \ll \mu$ and $y \ll fx$
\end{enumerate}
\end{definition}

\begin{proposition}
\label{prop:have_precise_supports}
$\gauss$, $\finstoch$ and $\borelstoch$ have precise supports. 
\end{proposition}
\begin{proof}
See appendix.
\end{proof}

\begin{theorem}
\label{thm:functoriality}
Let $\C$ be a Markov category that has conditionals and precise supports. Then $\sim$ is a functorial equivalence relation on $\obs(\C)$.
\end{theorem}
\begin{proof}
Let $(K,\psi,o) \sim (K',\psi',o') : I \obsto X$ and $(H,f,v) : X \obsto Y$ be any morphism. We need to show that
\begin{equation}
\label{eq:cond_fun_equiv}
(H \otimes K, (f \otimes \id_K)\psi, v \otimes o) \sim (H \otimes K', (f \otimes \id_{K'})\psi', v \otimes o')
\end{equation}
If $(K,f,o)$ fails, i.e. $o \not \ll \psi_K$, then by marginalization any composite must fail. But then the RHS fails too. 

Now assume that $(K,\psi,o)$ succeeds and $\psi|_{K}o = \psi'|_{K}o'$. We show that the success conditions on both sides are equivalent. That is because the following are equivalent
\begin{enumerate}
\item $v \otimes o \ll (f_H \otimes \id_K)\psi$
\item $o \ll \psi_K$ and $v \ll f_H\psi|_Ko$
\end{enumerate}
This is exactly the `precise supports' axiom, applied to $\mu = \psi_K$ and $g=f_H \circ \psi|_K$. Because 2) agrees on both sides of \eqref{eq:cond_fun_equiv}, so does 1). We are left with the case that \eqref{eq:cond_fun_equiv} succeeds, and need to show that
\[ [(f \otimes \id_K)\psi]|_{H \otimes K}(v \otimes o) = [(f \otimes \id_{K'})\psi']|_{H \otimes K'}(v \otimes o'). \]
We use a variant of the argument from \cite[11.11]{fritz} that double conditionals can be replaced by iterated conditionals. Consider the parameterized conditional 
\[ \beta = (f \circ \psi|_K)|_H : H \otimes K \to Y \]
then string diagram manipulation shows that $\beta$ has the universal property $\beta = [(f \otimes \id_K)\psi]|_{H \otimes K}$. By specialization \ref{prop:specialization}, it also has the property $\beta(\id_H \otimes o) = (f \circ \psi|_Ko)|_H$. Hence
\begin{align*}
[(f \otimes \id_K)\psi]|_{H \otimes K}(v \otimes o) &= \beta(\id_H \otimes o) \circ v \\
&= (f \circ \psi|_Ko)|_H \circ v \\
&= (f \circ \psi'|_{K'}o')|_H \circ v \\
&= [(f \otimes \id_{K'})\psi']|_{H \otimes K'}(v \otimes o)
\end{align*}
\end{proof}

We can spell out the equivalence $\approx$ as follows:
\begin{proposition}
\label{prop:simplified_cond}
We have $(K,f,o) \approx (K',f',o') : X \obsto Y$ if for all $\psi : I \to A \otimes X$, either
\begin{enumerate}
\item $o \ll f_K\psi_X$ and $o' \ll f'_{K'}\psi'_X$ and $[(\id_A \otimes f)\psi]|_K(o) = [(\id_A \otimes f')\psi']|_{K'}(o')$
\item $o \not \ll f_K\psi_X$ and $o' \not \ll f'_{K'}\psi'_X$
\end{enumerate}
\end{proposition}
The universal property of the conditional in question is
\[ \begin{tikzpicture}[scale=0.5]
	\begin{pgfonlayer}{nodelayer}
		\node [style=medium state] (19) at (-2, -1.75) {$\psi$};
		\node [style=none] (22) at (-2.5, -1.25) {};
		\node [style=none] (23) at (-1.5, -1.25) {};
		\node [style=morphism] (25) at (-1.5, 0) {$f$};
		\node [style=none] (26) at (-2.5, 3.25) {};
		\node [style=none] (27) at (-1.75, 0.5) {};
		\node [style=none] (28) at (-1.25, 0.5) {};
		\node [style=none] (29) at (-1.75, 3.25) {};
		\node [style=none] (30) at (-1.25, 3.25) {};
		\node [style=none] (31) at (0, 0) {$=$};
		\node [style=medium state] (32) at (1.75, -1.75) {$\psi$};
		\node [style=none] (33) at (1.25, -1.25) {};
		\node [style=none] (34) at (2.25, -1.25) {};
		\node [style=morphism] (35) at (2.25, 0) {$f$};
		\node [style=bn] (36) at (1.25, -0.5) {};
		\node [style=none] (37) at (2, 0.5) {};
		\node [style=none] (38) at (2.5, 0.5) {};
		\node [style=bn] (39) at (2, 1.25) {};
		\node [style=bn] (40) at (2.5, 2.25) {};
		\node [style=none] (41) at (1.5, 3.5) {};
		\node [style=none] (42) at (3.5, 3.5) {};
		\node [style=morphism] (43) at (1.5, 3.5) {$\phantom{Foo}$};
		\node [style=none] (44) at (1, 4.75) {};
		\node [style=none] (45) at (2, 4.75) {};
		\node [style=none] (46) at (3.5, 4.75) {};
		\node [style=none] (47) at (1, 4) {};
		\node [style=none] (48) at (2, 4) {};
		\node [style=none] (49) at (1, 5.5) {};
		\node [style=none] (50) at (1, 5.5) {$A$};
		\node [style=none] (51) at (2, 5.5) {$Y$};
		\node [style=observation] (52) at (3.5, 5.5) {$K$};
		\node [style=none] (55) at (-2.5, 4) {$A$};
		\node [style=none] (56) at (-1.75, 4) {$Y$};
		\node [style=observation] (57) at (-1, 4) {$K$};
	\end{pgfonlayer}
	\begin{pgfonlayer}{edgelayer}
		\draw (25) to (23.center);
		\draw (26.center) to (22.center);
		\draw (29.center) to (27.center);
		\draw [style=condition, in=270, out=90] (28.center) to (30.center);
		\draw (35) to (34.center);
		\draw (36) to (33.center);
		\draw (39) to (37.center);
		\draw [style=condition, in=270, out=90] (38.center) to (40);
		\draw [style=condition, bend left, looseness=1.25] (40) to (41.center);
		\draw [style=condition, bend right=45] (40) to (42.center);
		\draw [style=condition] (42.center) to (46.center);
		\draw (47.center) to (44.center);
		\draw (48.center) to (45.center);
	\end{pgfonlayer}
\end{tikzpicture} \]

We can show that isomorphic conditions are equivalent under the relation $\approx$.

\begin{proposition}[Isomorphic conditions]
\label{prop:isomorphic_conditions}
Let $(K,f,o) : X \obsto Y$ and $\alpha : K \cong K'$ be an isomorphism. Then
\[ (K,f,o) \approx (K',(\id_Y \otimes \alpha)f, \alpha o). \]
In programming terms $(k \eqo o) \approx (\alpha k \eqo \alpha o)$.
\end{proposition}
\begin{proof}
Let $\psi : I \to A \otimes X$. We first notice that $o \ll \psi_K$ if and only if $\alpha o \ll \alpha \psi_K$, so the success conditions coincide. It is now straightforward to check the universal property
\[ (\id_A \otimes f)\psi|_K = (\id_A \otimes ((\id_X \otimes \alpha)f))\psi|_{K'} \circ \alpha. \] 
This requires the fact that isomorphisms are deterministic, which holds in every Markov category with conditionals \cite[11.28]{fritz}. The proof works more generally if $\alpha$ is deterministic and split monic.
\end{proof}

We can now give the Cond construction:
\begin{definition}
Let $\C$ be a Markov category that has conditionals and precise supports. We define $\cond(\C)$ as the quotient \[ \cond(\C) = \obs(\C)/\!\!\approx \] This quotient is a CD category, and the functor $J : \C \to \cond(\C)$ preserves CD structure.
\end{definition}
\begin{proof}
We have checked functoriality of $\sim$ in \ref{thm:functoriality}, so by \ref{def:observational_quotient}, the quotient is symmetric premonoidal. It remains to show that the interchange laws holds, i.e. observations can be reordered. This follows from \ref{prop:isomorphic_conditions} because swap morphisms are iso.
\end{proof}

\begin{proposition}
\label{prop:good_properties}
By virtue of being a well-defined CD category, the program equations \eqref{eqn:commutativity} hold in the internal language of $\cond(\C)$. In particular, conditioning satisfies commutativity.
\end{proposition}

\subsection{Laws for Conditioning}

We derive some properties of $\cond(\C)$. We firstly notice that $J$ is faithful for common Markov categories. 
\begin{proposition}
\label{prop:faithful}
For $f, g : m \to n$, $J(f) \approx J(g)$ iff
\[ \forall \psi : I \to a \otimes m, (\id_a \otimes f)\psi = (\id_a \otimes g)\psi \]
In particular, $J$ is faithful for $\gauss$, $\finstoch$ and $\borelstoch$.
\end{proposition}
\begin{proof}
The proof is straightforward. This condition is stronger than equality on points: It implies that $f,g$ are almost surely equal with respect to all distributions. 
\end{proof}

\begin{proposition}[Closed terms]
\label{prop:closedterms}
There is a unique state $\bot_X : I \obsto X$ in $\cond(\C)$ that always fails, given by any $(K,\psi,o)$ with $o \not\ll \psi_K$. Any other state is equal to a conditioning-free posterior, namely $(K,\psi,o) \approx J(\psi|_K \circ o)$.
\end{proposition}

\begin{proposition}[Enforcing conditions]
\label{prop:exactly}
We have
\[ (X,\cpy_X,o) \approx (X,o \otimes \id_X,o) \]
\end{proposition}
This means conditions actually hold after we condition on them. In programming notation
\[ x \vdash (x \eqo o); x \approx (x \eqo o); o \] 
\begin{proof}
Let $\psi : I \to A \otimes X$; the success condition reads $o \ll \psi_X$ both cases. Now let $o \ll \psi_X$. We verify the properties
\begin{align*}
[(\id_A \otimes \cpy_X)\psi]|_X &= \langle \psi|_X, \id_X \rangle \\
[(\id_A \otimes o \otimes \id_X)\psi]|_X &= \psi|_X \otimes o
\end{align*}
and obtain
$ \langle \psi|_X, \id_X \rangle o = \psi|_X(o) \otimes o = (\psi|_X \otimes o)(o)$
from determinism of $o$.
\end{proof}
\section{Denotational semantics}
\label{sec:denotational}
We apply $\cond$ (\S\ref{sec:cond}) to give denotational semantics to our Gaussian language (\S\ref{sec:gaussian_language}), which we show to be fully abstract (Prop.~\ref{prop:full_abstraction}). One convenient feature is that we can use subtraction in $\gauss$ to condition on arbitrary expressions by observing a vanishing difference:

\begin{definition}
The Gaussian language embeds into the internal language of $\cond(\gauss)$, where $x \eq y$ is translated as $(x - y) \eqo 0$. A term $\vec x : R^m \vdash e : R^n$ denotes a morphism $\sem{e} : m \obsto n$.
\end{definition}

\begin{proposition}[Correctness]
If $(e,\psi) \red (e',\psi')$ then $\sem{e}\psi = \sem{e'}\psi'$. If $(e,\psi) \red \bot$ then $\sem{e} = \bot$.
\end{proposition}
\begin{proof}
We can faithfully interpret $\psi$ as a state in both $\gauss$ and $\cond(\gauss)$. If $x \vdash e$ and $(e,\psi) \red (e',\psi')$ then $e'$ has potentially allocated some fresh latent variables $x'$. We show that
\begin{equation}
\letin{x}{\psi} (x,\sem{e}) = \letin{(x,x')}{\psi'} (x,\sem{e'}).
\end{equation}
This notion is stable under reduction contexts. 

Let $C$ be a reduction context. 
Then
\begin{align*}
&\letin x \psi (x,\sem{C[e]}(x)) \\
&= \letin x \psi \letin y {\sem{e}(x)} (x,\sem{C}(x,y)) \\
&= \letin {(x,x')} {\psi'} \letin y {\sem{e'}(x,x')} (x,\sem{C}(x,y)) \\
&= \letin {(x,x')} {\psi'} (x,\sem{C[e']})
\end{align*}
Now for the redexes
\begin{enumerate}
\item The rules for $\mathrm{let}$ follow from the general axioms of value substitution in the internal language
\item For $\normal()$ we have $(\normal(), \psi) \red (x',\psi \otimes \mathcal N(0,1))$ and verify
\begin{align*}
&\letin x \psi (x,\sem{\normal()}) \\
&= \psi \otimes \mathcal N(0,1) \\
&= \letin {(x,x')} {\psi \otimes \mathcal N(0,1)} (x,\sem{x'})
\end{align*}
\item For conditioning, we have $(v\eq w,\psi) \red ((), \psi|_{v=w})$. We need to show
\begin{align*}
&\letin x \psi (x, \sem{v \eq w}) = \letin x {\psi|_{v=w}} (x,())
\end{align*}
Let $h=v-w$, then we need to the following morphisms are equivalent in $\cond(\gauss)$:
\[ \begin{tikzpicture}[scale=0.4]
	\begin{pgfonlayer}{nodelayer}
		\node [style=medium state] (0) at (4, -0.5) {$\psi|_{h=0}$};
		\node [style=none] (1) at (0, 0) {$\approx$};
		\node [style=medium state] (2) at (-4, -2) {$\psi$};
		\node [style=morphism] (3) at (-2.75, 1) {$h$};
		\node [style=none] (4) at (-5.25, 0.75) {};
		\node [style=none] (7) at (-4, -1.5) {};
		\node [style=none] (8) at (4, 0) {};
		\node [style=none] (9) at (4, 4.5) {};
		\node [style=bn] (12) at (-4, -0.5) {};
		\node [style=none] (13) at (-2.75, 3.25) {};
		\node [style=none] (14) at (-5.25, 4.5) {};
		\node [style=none] (15) at (-2.75, 4) {};
		\node [style=observation] (16) at (-2.75, 4) {$0$};
	\end{pgfonlayer}
	\begin{pgfonlayer}{edgelayer}
		\draw (9.center) to (8.center);
		\draw (12) to (7.center);
		\draw [bend left=45] (12) to (4.center);
		\draw [in=-90, out=0, looseness=0.75] (12) to (3);
		\draw [style=condition] (3) to (13.center);
		\draw (14.center) to (4.center);
	\end{pgfonlayer}
\end{tikzpicture} \]
Applying \ref{prop:closedterms} to the left-hand side requires us to compute the conditional $\langle \id, h\rangle\psi|_2 \circ 0$, which is exactly how $\psi|_{h=0}$ is defined.
\end{enumerate}
\vspace{-12pt}
\end{proof}

\begin{proposition}[Full abstraction]
\label{prop:full_abstraction}
$\sem{e_1} = \sem{e_2}$ if and only if $e_1 \approx e_2$ (where $\approx$ is contextual equivalence, Def.~\ref{def:ctxequiv}).
\end{proposition}
\begin{proof}
For $\Rightarrow$, let $K[-]$ be a closed context. Because $\sem{-}$ is compositional, we obtain $\sem{K[e_1]} = \sem{K[e_2]}$. If both succeed, we have reductions $(K[e_i], !) \red^* (v_i,\psi_i)$ and by correctness $v_1\psi_1 = \sem{K[e_1]} = \sem{K[e_2]} = v_2\psi_2$ as desired. If $\sem{K[e_1]} = \sem{K[e_2]} = \bot$ then both $(K[e_i], !) \red^* \bot$. 

For $\Leftarrow$, we note that $\cond$ quotients by contextual equivalence, but all Gaussian contexts are definable in the language. 
\end{proof}
\section{Equational theory}
\label{sec:equational_theory}

\newcommand{\ajname}{\mathrel{\vdash\!\!\!_{\mathsf{a}}}}
\newcommand{\cjname}{\mathrel{\vdash\!\!\!_{\mathsf{c}}}}
\newcommand{\aj}[3]{#1\ajname #2:#3}
\newcommand{\cj}[3]{#1\cjname #2:#3}
\newcommand{\realtype}{\rv}
\newcommand{\abort}\bot
\newcommand{\returnname}{\mathrm{return}}
\newcommand{\return}[1]{\returnname(#1)}
\newcommand{\nor}[2]{\nu #1.#2}
We now give an explicit presentation of the equality between programs in the Gaussian language (\S\ref{sec:equational_theory_axioms}). We demonstrate the strength of the axioms by using them to characterize normal forms for various fragments of the language (\S\ref{sec:nf}). Besides an axiomatization of program equality, this can also be regarded in other equivalent ways, such as a presentation of a PROP by generators and relations, or as a presentation of a strong monad by algebraic effects, or as a presentation of a Freyd category. But we approach from the programming language perspective. 

\subsection{Presentation}\label{sec:equational_theory_axioms}
We use the following fragment of the language from \S\ref{sec:gaussian_language}. The reader may find it helpful to think of this as a normal form for the language modulo associativity of `let'. This fragment has the following modifications: only variables of type $\realtype$ are allowed in the typing context~$\Gamma$; we have an explicit command for failure~($\bot$); we separate the typing judgement in two: judgements for expressions of affine algebra~$\ajname$ and for general computational expressions~$\cjname$; we have an explicit coercion `$\returnname$' between them for clarity. 
\[ \infer{\aj{\Gamma,x : \realtype,\Gamma'} x \realtype}{}  
\quad\ 
\infer{\aj{\Gamma}{s+t}{\realtype}}{\aj\Gamma s\realtype \quad\aj\Gamma t \realtype}
\quad\ 
\infer{\aj\Gamma {\alpha \cdot t} \realtype}{\aj\Gamma t\realtype}
\]\[
\infer{\aj\Gamma {\underline{\beta} }\realtype}{}
\qquad
\infer{\cj\Gamma {\letin x {\normal()} t}\realtype^n }{\Gamma, x : \realtype \vdash t : \realtype^n} 
\]\[
\infer{\cj\Gamma  {(s \eq t);u} {\realtype^n}}{\aj\Gamma s \realtype \quad \aj \Gamma t \realtype\quad \cj\Gamma u {\realtype^n}}
\]
\[
\infer{\cj\Gamma {\return{t_1,\dots,t_n}} {\realtype^n}}{\aj\Gamma {t_1} \realtype\quad\dots\quad \aj\Gamma {t_n} \realtype}
\qquad
\infer{\cj\Gamma \abort {\realtype^n}}{}
\]

There is no general sequencing construct, but we can combine expressions using the following substitution constructions, whose well-typedness is derivable. 
\[
\infer{\cj{\Gamma,\Gamma'} {t[s/x]}{\realtype^n} }{\cj{\Gamma,x:\realtype,\Gamma'}t {\realtype^n}\quad \aj {\Gamma,\Gamma'} s \realtype}
\]\[
    \infer{\cj{\Gamma} {t[u/\returnname]}{\realtype^n} }{\cj{\Gamma}t {\realtype^m}\quad \cj {\Gamma,x_1,\dots,x_m:\realtype} {x_1,\dots,x_n.u} {\realtype^n}}
    \]
    \newcommand{\ceq}[4]{\cj{#1}{#3 \equiv #4}{\realtype^{#2}} }
    \newcommand{\ceqal}[4]{\cj{#1&}{#3 \equiv #4}{\realtype^{#2}} }
    \newcommand{\retname}{\mathrm{r}}
    \newcommand{\ret}[1]{\retname[#1]}
    In the second form we replace the return statement of an expression with another expression, capturing variables appropriately. The precise definition of this hereditary substitution is standard in logical frameworks (e.g.~\cite{adams-tf},\cite{staton:pl}), for example:
\begin{multline*}
    \hspace{-3mm}\big(\letin x {\normal()} {\return{x+3}}
    \big)[^{a.a\eq 4;\return a}\!/\!_\returnname]\\
    =\ \letin x {\normal()} {(x+3)\eq 4;\return{x+3}}\end{multline*}
For brevity we now write $\nor x t$ for $\letin x {\normal()} t$, $\retname$ for $\returnname$ and drop `$;$' when unambiguous. We axiomatize equality by closing the following axioms under the two forms of substitution and also congruence. Now the syntax has the appearance of a second order algebraic theory, similar to the familiar presentations of $\lambda$-calculus or predicate logic. \\

The theory is parameterized over an underlying theory of values, which is affine algebra. The type $\realtype$ has the structure of a \emph{pointed vector space}, which obeys the usual axioms of vector spaces plus constant symbols $(\underline{\beta})_{\beta \in \R}$ subject to 
\[ \alpha \cdot \underline{\beta} = \underline{\alpha\beta}, \underline{\alpha} + \underline{\beta} = \underline{\alpha + \beta} \]
Terms modulo equations are affine functions. The category theorist will recognize the category $\mathbb A=\gauss_\det$ as the Lawvere theory of pointed vector spaces.

The following axioms characterize the conditioning-free fragment of the language, that is, Gaussian probability
\begin{align}
\ceqal {} 0 {\nor x {\ret{}}} {\ret{}} \label{ax:discard} \tag{DISC} \\
\ceqal {} n {\nor{\vec x} {\ret{U \vec x}}} {\nor {\vec x} {\ret {\vec x}}} \text{ if $U$ orthogonal} \label{ax:orth} \tag{ORTH}
\end{align}

The following are commutativity axioms for conditioning
\begin{align}
\ceqal {\hspace{-4mm}a,b,c,d }0 {(a \eq b)(c \eq d)\ret{}}{(c \eq d)(a \eq b)\ret{}} \label{ax:cond_comm} \tag{C1} \ \\ 
\ceqal{a,b }1 {(a \eq b);\nor x{\ret{x}}} {\nor x{(a \eq b)\ret{x}}}  \label{ax:nu_comm_cond} \tag{C2} \\\
\ceqal {a,b} n {(a \eq b)\bot} \bot \label{ax:fail_absorb} \tag{C3} \
\end{align}
while following encode specific properties of $(\eq)$
\begin{align}
\ceqal a 0 {(a \eq a)\ret{}}{\ret{}} \label{ax:cond_tauto} \tag{TAUT} \ \\ 
\ceqal{}0 {(\underline{0} \eq \underline{1})\ret{}} {\bot} \label{ax:fail_incons} \tag{FAIL} \ \\
\ceqal{a,b} 1 {(a \eq b)\ret a}{(a \eq b)\ret b} \label{ax:cond_subst} \tag{SUBS} \ \\
\ceqal {}1{\nor x{(x \eq \underline c)\ret{x}}} {\ret{\underline c}} \label{ax:cond_init}  \tag{INIT}
\end{align}
Lastly, we add the special congruence scheme
\begin{align}
\ceq{\Gamma} 0 {(s \eq t)\ret {}} {(s' \eq t')\ret {}} \label{ax:cong} \tag{CONG} \
\end{align}
whenever $(s=t)$ and $(s'=t')$ are interderivable equations over $\Gamma$ in the theory of pointed vector spaces. \\

Axioms~\eqref{ax:discard} and~\eqref{ax:orth} completely axiomatize the fragment of the language without conditioning (Prop.~\ref{prop:gauss_complete}). Axioms~\eqref{ax:cond_comm}-\eqref{ax:fail_absorb} describe dataflow -- all the operations distribute over each other. The reader should focus on the remaining five axioms \eqref{ax:cond_tauto}-\eqref{ax:cong}, which are specific to conditioning. It is intended that they are straightforward and intuitive. 

\subsection{Normal forms}
\label{sec:nf}

\begin{proposition}
\label{prop:gauss_complete}
Axioms \eqref{ax:discard}-\eqref{ax:orth} are complete for $\gauss$. That is, conditioning-free terms $\vec x : \rv^n \vdash u,v : \rv^n$ denote the same morphism in $\gauss$ if and only if $\vec x \vdash u \equiv v$ is derivable from the axioms.
\end{proposition}
\begin{proof}
The axioms are clearly validated in $\gauss$; probability is discardable and independent standard normal Gaussians are invariant under orthogonal transformations. Note that $\nu$ commutes with itself because permutation matrices are orthogonal. 

It is curious that these laws completely characterize Gaussians: Any term normalizes to the form $\nu \vec z.\ret{A \vec x + B \vec z + \vec c}$, denoting the map $(A,\vec c,BB^T)$ in $\gauss$. Consider some other term $\nu \vec w.\varphi[A'\vec x + B'\vec w + \vec c']$ that has the same denotation. By \eqref{ax:discard}, we can without loss of generality assume that $\vec z$ and $\vec w$ have the same dimension. The condition $(A,c,BB^T)=(A',c',B'(B')^T)$ implies $A=A',\vec c= \vec c'$. By \ref{prop:orth_equivalence} there is an orthogonal matrix $U$ such that $B' = BU$. So the two terms are equated under \eqref{ax:orth}.
\end{proof}

\begin{example}
\label{ex:alg_diagonal}
\[ \nu x.\nu y. \ret{x+y} \equiv \nu y.\ret{\sqrt 2 \cdot y} \]
\end{example}
\begin{proof}
Let $s=1/\sqrt 2$, then the matrix $U = \begin{pmatrix} s & s \\ -s & s \end{pmatrix}$ is orthogonal. Thus
\begin{align*}
\nu x.\nu y. \ret{x+y} 
&\equiv \nu x.\nu y. \ret{(sx + sy) + (-sx + sy)} \\
&\equiv \nu x.\nu y. \ret{\sqrt 2 y} \\
&\equiv \nu y. \ret{\sqrt 2 y} \
\end{align*} 
where we apply \eqref{ax:orth}, affine algebra and \eqref{ax:discard}.
\end{proof}

We proceed to showing the consistency of the axioms for conditioning.

\begin{proposition}
Axioms \eqref{ax:discard}-\eqref{ax:cong} are valid in $\cond(\gauss)$
\end{proposition}
\begin{proof}
Sketch. The commutation properties are straightforward from string diagram manipulation. 
\begin{description}
	\item[{\eqref{ax:cond_subst}}]~Write $a=b+(a-b)$; by \ref{prop:exactly}, once we condition $a-b \eqs 0$, we have $a=b$.
	\item[{\eqref{ax:cond_init}}]~By \ref{prop:closedterms}, noting that $c \ll \mathcal N(0,1)$
	\item[{\eqref{ax:fail_incons}}]~By \ref{prop:closedterms}, because $0 \not \ll 1$
	\item[{\eqref{ax:cong}}]~This follows from \ref{prop:isomorphic_conditions}, because over $\A$, equivalent scalar equations are nonzero multiples of each other. Still, this is very surprising axiom scheme, which is substantially generalized in \ref{prop:equivalent_equations}.
\end{description} 
\end{proof}

For the remainder of this section, we will show how to use the theory to derive normal forms for conditioning programs.

\begin{proposition}
\label{prop:elementary}
Elementary row operations are valid on systems of conditions. In particular, if $S$ is an invertible matrix then
\[ (A\vec x \eq \vec b)\ret{} \equiv (SA\vec x \eq S\vec b)\ret{} \]
\end{proposition}
\begin{proof}
Reordering and scaling of equations is \eqref{ax:cond_comm}, \eqref{ax:cong}. For summation, i.e.
\[ (s \eq t)(u \eq v)\ret{} \equiv (s \eq t)(u + s \eq v + t)\ret{} \]
instantiate \eqref{ax:cond_subst} with $^{(u + x \eq v + t)\ret{}}/_{\ret{x}}$. Now use the fact that applying any invertible matrix on the left can be decomposed into elementary row operations.
\end{proof}

\begin{corollary}
\label{prop:equivalent_equations}
If $A\vec x = \vec c$ and $B\vec x = \vec d$ are linear systems of equations with the same solution space, then 
\[ (A\vec x \eq \vec c)\ret{} \equiv (B \vec x \eq \vec d)\ret{} \]
is derivable. 
\end{corollary}
This generalizes \eqref{ax:cong} to systems of conditions.
\begin{proof}
If the systems are consistent, then they are isomorphic by \ref{prop:affine_equivalence} and we use the previous proposition. If they are inconsistent, we can derive $(0\eq 1)$ and use \eqref{ax:fail_incons},\eqref{ax:fail_absorb} to equate them to $\bot$.
\end{proof}

We give a normal form for closed terms.
\begin{theorem}
\label{prop:normalization}
Any closed term can be brought into the form $\nu \vec z. \ret{A\vec z + \vec c}$ or $\bot$. The matrix $AA^T$ is uniquely determined.
\end{theorem}
This is the algebraic analogue of \ref{prop:closedterms}.
\begin{proof}
By commutativity, we bring the term into the form
\[ \nu \vec z. (A \vec z \eq \vec b)\ret{D\vec z+\vec d} \]
By \ref{prop:basechange}, we can find invertible matrices $S,T$ such that
\[ SAT^{-1} = \begin{pmatrix} I_r & 0 \\ 0 & 0 \end{pmatrix} \]
and $T$ is orthogonal. Using the orthogonal coordinate change $\vec w=T\vec z$ and \ref{prop:equivalent_equations}, the equations take the form
\[ \nu \vec w.(SAT^{-1}\vec w \eq S\vec b)\ret{DT^{-1}\vec w + \vec d} \]
This simplifies to
\begin{align*}
\nu \vec w. (\vec w_{1:r} \eq \vec c_{1:r})(0 \eq \vec c_{r+1:n}) \ret{DT^{-1}\vec w + \vec d}
\end{align*}
where $\vec c = S\vec b$. We can process the first block of conditions with \eqref{ax:cond_init}. The conditions $(\underline{0} \eq c_i)$ can either be discarded by \eqref{ax:cond_tauto} if $c_i = \underline{0}$ for all $i=r+1,\ldots,n$, or fail by \eqref{ax:fail_incons} otherwise. We arrive at a conditioning-free term.
\end{proof}

\begin{example}
\[ \nu x.\nu y.(x \eq y) \ret{x,y} \equiv \nu x.\ret{sx,sx} \]
where $s=1/\sqrt 2$.
\end{example}
\begin{proof}
We use again the unitary matrix $U$ from Example \ref{ex:alg_diagonal}
\begin{align*}
\nu x.\nu y.(x \eq y);\ret{x,y}
&\equiv \nu x.\nu y.(sx + sy \eq -\!sx + sy);\\
        &\quad\quad \ret{sx + sy, -sx + sy} \\
&\equiv \nu x.\nu y.(x \eq 0)\ret{sx + sy, -sx + sy} \\
&\equiv \nu y.\ret{sy,sy}
\end{align*}
where we apply \eqref{ax:orth}, affine algebra and \eqref{ax:cond_init}.
\end{proof}

Lastly, we give a normal form for conditioning effects.

\begin{theorem}[Normal forms]
\label{prop:normalform_uniqueness}
Every term $\vec x : \rv^n \vdash u : \rv^0$ can either be brought into the form $\bot$ or
\begin{equation} \nu \vec z.A \vec x =:= B\vec z + \vec c \label{eq:normalform} \end{equation}
where $A \in \R^{r \times n}$ is in reduced echelon form with no zero rows. The values of $A$, $\vec c$ and $BB^T$ are uniquely determined.
\end{theorem}
\begin{proof}
Through the commutativity axioms, we can bring $u$ into the form
\[ \nu \vec z.A \vec x =:= B\vec z + \vec c \]
for general $A$. Let $S$ be an invertible matrix that turns $A$ into reduced row echelon form, and apply it to the condition via \ref{prop:elementary}. The zero columns don't involve $\vec x$, so we use \ref{prop:normalization} to evaluate the condition involving $\vec z$ separately. We either obtain $\bot$ or the desired form \eqref{eq:normalform}

For uniqueness, we consider the term's denotation $(A\vec x \eq \eta) : n \obsto 0$ in $\cond(\gauss)$, where $\eta = \mathcal N(\vec c,BB^T)$. We must show that $A$ and $\eta$ can be reconstructed from the observational behavior of the denotation. The proof given in the appendix \ref{prop:uniqueness_proof}.
\end{proof}

\section{Context, related work, and outlook}
\subsection{Symbolic disintegration and paradoxes}
\label{sec:paradoxes}
Our line of work can be regarded as a synthetic and axiomatic counterpart of the symbolic disintegration of Shan and Ramsey~\cite{shan:symbolic-disintegration}. (See also~\cite{psi,delayed_sampling,shan:disint,shan}.) 
That work provides in particular verified program transformations to convert an arbitrary probabilistic program of type $\realtype\otimes \tau$ 
to an equivalent one that is of the form $\letin x {\mathrm{lebesgue}()} {\letin y M {(x,y)}}$. 
So now exact conditioning $x\eq o$ can be carried out by substituting $o$ for $x$ in $M$. 
We emphasize the similarity with the definition of conditionals in Markov categories, as well as the role that coordinate transformations play in both our work (\S\ref{sec:equational_theory}) and \cite{shan:symbolic-disintegration}. One language novelty in our work is that exact conditioning is a first-class construct, as opposed to a whole-program transformation, in our language, which makes the consistency of exact conditioning  more apparent.

Consistency is a fundamental concern for exact conditioning. \emph{Borel's paradox} is an example of an inconsistency that arises if one is careless with exact conditioning \cite[Ch.~15]{jaynes}, \cite[\S3.3]{jacobs:paradoxes}: It arises when naively substituting equivalent equations within $(\eq)$. For example, the equation $x - y = 0$ is equivalent to $x/y = 1$ over the (nonzero) real numbers. Yet, in an extension of our language with division, the following programs are not contextually equivalent:
\begin{equation*}
\begin{array}{l}
\mlstinline{x = normal(0,1)} \\
\mlstinline{y = normal(0,1)} \\
\mlstinline{x-y =:= 0} \\
\end{array}
\not \equiv
\begin{array}{l}
\mlstinline{x = normal(0,1)} \\
\mlstinline{y = normal(0,1)} \\
\mlstinline{x/y =:= 1}
\end{array}
\end{equation*}
On the left, the resulting variable $x$ has distribution $\mathcal N(0,0.5)$ while on the right, $x$ can be shown to have density $|x|e^{-x^2}$ \cite{expect_the_unexpected, shan:symbolic-disintegration}. In our work, Borel's paradox finds a type-theoretic resolution: Conditioning is presented abstractly as an algebraic effect, so the expressions $(s \eq t) : \unit$ and $(s == t) : \mathrm{bool}$ have a different formal status and can no longer be confused. They are related explicitly through axioms like \eqref{ax:cond_subst}, and special laws for simplifying conditions are given in \eqref{ax:cong}, \ref{prop:equivalent_equations}. By \ref{prop:isomorphic_conditions}, we can always substitute conditions which are formally isomorphic, but $x-y \eq 0$ and $x/y \eq 1$ are not isomorphic conditions in this sense. For the special case of Gaussian probability, we proved that equivalent affine equations are automatically isomorphic, making it very easy to avoid Borel's paradox in this restricted setting (Prop.~\ref{prop:equivalent_equations}). To include the non-example above, our language needs a nonlinear operation like $(/)$. If beyond that we introduced equality testing to the language, difference between equations and conditions would become even more apparent. The \emph{equation} $x-y=0$ is obviously equivalent to the equation $(x == y) = \mathrm{true}$, but the \emph{condition} $(x == y) \eq \mathrm{true}$ would cause the whole program to fail, since measure-theoretically, $(x==y)$ is the same as $\mathrm{false}$. 

This also suggests a tradeoff between expressivity of the language and well-behavedness of conditioning. On this subject, Shan and Ramsey~\cite{shan:symbolic-disintegration} wrote:
\begin{quote}{\it The [measure-theoretic] definition of disintegration allows latitude that our disintegrator does not take: When we disintegrate $\xi = \Lambda \otimes \kappa$, the output $\kappa$ is unique only almost everywhere --- $\kappa x$ may return an arbitrary measure at, for example, any finite set of~$x$'s. But our disintegrator never invents an arbitrary measure at any point. The mathematical definition of disintegration is therefore a bit too loose to describe what our disintegrator actually does. How to describe our disintegrator by a tighter class of ``well-behaved disintegrations'' is a question for future research. In particular, the notion of continuous disintegrations~\cite{afr2016} is too tight, because depending on the input term, our disintegrator does not always return a continuous disintegration, even if one exists.}\end{quote}
In this paper we have tackled this research problem: a notion of ``well-behaved disintegrations'' is given by a Markov category with precise supports. The most comprehensive category $\borelstoch$ admits conditioning only on events of positive probability (\ref{prop:stoch_support}). The smaller category $\gauss$ features a better notion of support and an interesting theory of conditioning. Studying Markov categories of different degrees of specialization helps navigating the tradeoff. 
Once in the synthetic setting of a Markov category $\C$ with precise supports, the program transformations of \cite{shan:symbolic-disintegration} are all valid in $\cond(\C)$, and the Markov conditioning property (Def.~\ref{def:conditional}) exactly matches the correctness criterion for symbolic disintegration. 

\subsection{Other directions}
Once a foundation is in algebraic or categorical form, it is easy to make connections to and draw inspiration from a variety of other work. 

The $\obs$ construction (\S\ref{def:obs}) that we considered here is reminiscent of the lens construction~\cite{lenses} and the Oles construction~\cite{hermida-tennent}. These have recently been applied to probability theory \cite{smithe:bayesian} and quantum theory \cite{huotstatonqpl}. The details and intuitions are different, but a deeper connection or generalization may be profitable in the future. 

Algebraic presentations of probability theories and conjugate-prior relationships have been explored in \cite{staton:betabernoulli}. Furthermore, the concept of exact conditioning is reminiscent of unification in Prolog-style logic programming. Our presentation in Section~\ref{sec:equational_theory} is partly inspired by the algebraic presentation of predicate logic in~\cite{staton:pl}, which has a similar signature and axioms. One technical difference is that in logic programming, $\exists a.\ret a\ \equiv\ \exists a.\exists b.(a\eq b)\ret{a}$ holds whereas 
here we have $\nu a.\ret {a}\ \equiv \ \nu a.\nu b.(a\eq b)\ret {(1/\sqrt 2) a }$, so things are more quantitative here. By collapsing Gaussians to their supports (forgetting mean and covariance), we do in fact obtain a model of unification.

Logic programming is also closely related to relational programming, and we note that our presentation is reminiscent of presentations of categories of linear relations~\cite{baez:control,bsz,pbsz}. 

On the semantic side, we recall that presheaf categories have been used as a foundation for logic programming~\cite{kinoshita-power}. Our axiomatization can be regarded as the presentation of a monad on the category $[\A^{\mathrm{op}}, \mathbf{Set}]$, via~\cite{staton:pl}, where $\A$ is the category of finite dimensional affine spaces discussed in \S\ref{sec:equational_theory}.

\emph{Probabilistic logic programming}~\cite{problog} supports both logic variables as well as random variables within a common formalism. We have not considered logic variables in this article, but a challenge for future work is to bring the ideas of exact conditioning closer to the ideas of unification, both practically and in terms of the semantics. We wonder if it is possible to think of $\exists$ as an idealized ``flat'' prior.

\section*{Acknowledgment}
It has been helpful to discuss this work with many people, including Tobias Fritz, Tom\'{a}\v{s} Gonda, Mathieu Huot, Ohad Kammar and Paolo Perrone. Research supported by a Royal Society University Research Fellowship and the ERC BLAST grant.

\bibliographystyle{IEEEtranS}

\bibliography{main}

\section*{Appendix}

\subsection{Markov categories}

Here, we spell out some details for the notions of $\ll$ and precise supports.

\begin{proof}[Proof of Proposition \ref{prop:gaussian_supports}]
$\gauss$ faithfully embeds into $\borelstoch$, that is any Gaussian morphism $m \to n$ can be seen as a measurable map $\R^m \to \mathcal G(\R^n)$, where $\mathcal G$ denotes the Giry monad. By \cite[13.3]{fritz}, we have $f=_\mu g$ if $f(x)=g(x)$ in $\mathcal G(\R^n)$ for $\mu$-almost all $x$. If $\mu$ is a Gaussian distribution with support $S$, then $\mu$ is equivalent to the Lebesgue measure on $S$. Because $f,g$ are continuous as maps into $\mathcal G(\R^m)$, $f(x) = g(x)$ on almost all $x \in S$ implies $f=g$ everywhere on $S$. 

For the second part, let $S_\mu,S_\nu$ denote the supports of $\mu$ and $\nu$, and let $x \in S_\mu \setminus S_\nu$. Then we can find two affine functions $f,g$ which agree on $S_\nu$ but $f(x) \neq g(x)$. Then $f=_\nu g$ but not $f=_\mu g$, hence $\mu \not\ll \nu$.
\end{proof}

\begin{proposition}
\label{prop:stoch_support}
In $\finstoch$, we have $x \ll \mu$ if $\mu(x) > 0$. In $\borelstoch$, we have $x \ll \mu$ if $\mu(\{x\}) > 0$.
\end{proposition}
This gives the correct intuition of support for $\finstoch$. For $\borelstoch$, this is an overly rigid notion of support which may contradict our intuition. For example the standard-normal distribution $\mathcal N(0,1)$ has support $\R$ in $\gauss$, but $\emptyset$ in $\borelstoch$.
\begin{proof}
The arguments follow from \cite[13.2,13.3]{fritz}. For $\borelstoch$, let $\mu(\{x_0\})=0$ and consider the measurable functions $f(x)=\mathbb I_{\{x_0\}}(x), g(x)=0$. Then $f=_\mu g$, yet $f(x_0) \neq g(x_0)$ showing $x_0 \not\ll \mu$.
\end{proof}

\begin{proof}[Proof of Proposition \ref{prop:have_precise_supports}]
For Gauss, this follows from the characterization of $\ll$ in \ref{prop:gaussian_supports}. Let $\mu$ have support $S$ and $f(x)=Ax+\mathcal N(b,\Sigma)$. Let $T$ be the support of $\mathcal N(b,\Sigma)$. The support of $\langle \id, f \rangle\mu$ is the image space $\{ (x,Ax+c) : x \in S, c \in T \}$. Hence $(x,y) \ll \langle \id, f \rangle\mu$ iff $x \ll \mu$ and $y \ll fx$. 
	
For $\finstoch$, an outcome $(x,y)$ has positive probability under $\langle \id, f \rangle\mu$ iff $x$ has positive probability under $\mu$, and $y$ has positive probability under $f(-|x)$.
	
For $\borelstoch$, the measure $\psi = \langle \id, f \rangle \mu$ is given by
	\[ \psi(A \times B) = \int_{x \in A} f(B|x) \mu(\mathrm dx) \]
Hence $\psi(\{(x_0,y_0)\}) = f(\{y_0\}|x)\mu(\{x\})$, which is positive exactly if $\mu(\{x_0\}) > 0$ and $f(\{y_0\}|x)>0$.
\end{proof}

A note on the definition of `precise supports': The expression $\langle \id,f \rangle\mu$ is an analogue of the graph of $f$. We wonder about its single-valuedness. If $x \otimes y \ll \langle \id,f \rangle\mu$ then we always have $x \ll \mu$ and $y \ll f\mu$. We ask that $y$ doesn't lie in the pushforward of any old sample of $\mu$, but precisely of $fx$. 
This is certainly a natural property to demand, but also very specifically tailored towards the application in \ref{thm:functoriality}. We expect `precise supports' to arise as an instance of some more encompassing axiom. \\

\subsection{Linear algebra}

The following facts from linear algebra are useful to recall and get used throughout.

\begin{proposition}
\label{prop:basechange}
Let $A \in \R^{m \times n}$, then there are invertible matrices $S,T$ such that
\[
SAT^{-1} = \begin{pmatrix} I_r & 0 \\ 0 & 0 \end{pmatrix}
\]
where $r = \mathrm{rank}(A)$. Furthermore, $T$ can taken to be orthogonal.
\end{proposition}
\begin{proof}
Take a singular-value decomposition (SVD) $A=UDV^T$, let $T=V$ and use create $S$ from $U^T$ by rescaling the appropriate columns.
\end{proof}

\begin{proposition}[Row equivalence]
\label{prop:row_equivalence}
Two matrices $A,B \in \R^{m \times n}$ are called \emph{row equivalent} if the following equivalent conditions hold
\begin{enumerate}
\item for all $x \in \R^n$, $Ax = 0 \Leftrightarrow Bx = 0$
\item $A$ and $B$ have the same row space
\item there is an invertible matrix $S$ such that $A=SB$
\end{enumerate}
Unique representatives of row equivalence classes are matrices in \emph{reduced row echelon form}.
\end{proposition}

\begin{corollary}
\label{prop:affine_equivalence}
Let $A,B \in \R^{m \times n}$ and let $Ax = c$ and $Bx = d$ be \emph{consistent} systems of linear equations that have the same solution space. There is an invertible matrix $S$ such that $B=SA$ and $d=Sc$.
\end{corollary}

\begin{proposition}[Column equivalence]
\label{prop:column_equivalence}
For matrices $A,B \in \R^{m \times n}$, the following are equivalent
\begin{enumerate}
\item $A$ and $B$ have the same column space
\item there is an invertible matrix $T$ such that $A=BT$.
\end{enumerate}
\end{proposition}

\begin{proposition}
\label{prop:orth_equivalence}
For matrices $A,B \in \R^{m \times n}$, the following are equivalent
\begin{enumerate}
\item $AA^T = BB^T$
\item there is an orthogonal matrix $U$ such that $A=BU$.
\end{enumerate}
\end{proposition}
\begin{proof}
This is a known fact, but we sketch a proof for lack of reference. In the construction of the SVD $A=UDV^T$, we can choose $U$ and $D$ depending on $AA^T$ alone. It follows that the same matrices work for $B$, giving SVDs $A=UDV^T, B=UDW^T$. Then $A=B(WV^T)$ as claimed.
\end{proof}

\subsection{Normal forms}

We present a proof of the uniqueness of normal forms for conditioning morphisms. Some preliminary facts:

\begin{proposition}
\label{prop:gauss_observe}
Let $X \sim \mathcal N(\mu_X,\Sigma_X)$ and $Y \sim \mathcal N(\mu_Y,\Sigma_Y)$ be independent. Then $X|(X=Y)$ has distribution $\mathcal N(\bar \mu,\bar \Sigma)$ given by
\begin{align*}
\bar \mu &= \mu_X + \Sigma_X(\Sigma_X+\Sigma_Y)^+(\mu_Y - \mu_X) \\
\bar \Sigma &= \Sigma_X - \Sigma_X(\Sigma_X+\Sigma_Y)^+\Sigma_X
\end{align*}
\end{proposition}
In programming terms, this is written
\[ \letin x {N(\mu_X,\Sigma_X)} (x \eq N(\mu_Y,\Sigma_Y)); \return x \]
and corresponds to the \lstinline|observe| statement from the introduction
\[ \mlstinline{x = normal($\mu_X$,$\Sigma_X$); observe(normal($\mu_Y$,$\Sigma_Y$), x)} \]

\begin{corollary}
\label{prop:gauss_1d}
No 1-dimensional observe statement leaves the prior $\mathcal N(0,1)$ unchanged.
\end{corollary}
\begin{proof}
Conditioning decreases variance. If we observe from $\mathcal N(\mu,\sigma^2)$, the variance of the posterior is 
\[ 1 - (1+\sigma^2)^{-1} < 1. \]
\end{proof}

\begin{proposition}
\label{prop:uniqueness_proof}
Consider a morphism $\kappa : n \obsto 0$ in $\cond(\gauss)$ given by
\begin{equation}
\label{eq:cond_normalform}
\kappa(x) = (Ax \eq \mathcal \eta)
\end{equation}
where $A \in \R^{r \times n}$ is in reduced row echelon form with no zero rows, and $\eta \in \gauss(0,r)$. Then the matrix $A$ and distribution $\eta$ are uniquely determined.
\end{proposition}
\begin{proof}
We will probe $\kappa$ by applying the condition \ref{eq:cond_normalform} to different priors $\psi \in \gauss(0,n)$, giving either a result $\psi' \in \gauss(0,n)$ or $\bot$.

Let $S \subseteq \R^r$ be the support of $\eta$ and $W = \{ x \in R^n : Ax \in S \}$. We can recover $W$ from observational behavior, because for deterministic priors $\psi = x_0$, we have $\psi' \neq \bot$ iff $x_0 \in W$. We have $\kappa = \bot$ iff $W=\emptyset$. Assume $W$ is nonempty now.

Next, we can identify the nullspace $K$ of $A$ by considering subspaces along which no conditioning updates happens. Call an affine subspace $V \subseteq \R^n$ \emph{vacuous} if for all $\psi \ll V$ we have $\psi'=\psi$. Any such $V$ must be contained in $W$. We claim that every maximal vacuous subspace is of the form $K+x_0$ where $x_0 \in W$: 

Every space of the form $K+x_0$ is clearly vacuous: If $\psi \ll K$ then the condition \eqref{eq:cond_normalform} becomes constant as $Ax_0 \eq \eta$. Because by assumption $Ax_0 \in S$, this condition is vacuous and can be discarded without effect.

Let $V$ be any vacuous subspace and $x_0 \in V$. We show $V \subseteq x_0 + K$: Assume there is any other $x_1 \in W$ such that $x_1-x_0 \notin K$ and consider the 1-dimensional prior 
\[ t \sim \mathcal N(0,1), \quad x=x_0 + t(x_1 - x_0) \]
Let $d=A(x_1-x_0) \neq 0$ and find an invertible matrix $T$ such that $Td = (1,0,\ldots,0)^T$. The condition becomes
\[ (t,0,\ldots,0) \eq T\eta-TAx_0. \]
All but the first equations do not involve $t$. By commutativity, they can be computed independently, resulting in an updated right-hand side and a 1-dimensional condition $t =:= \eta'$ with $\eta'$ either a Gaussian or $\bot$. By \ref{prop:gauss_1d}, such a condition cannot leave the prior $\mathcal N(0,1)$ unchanged, contradicting vacuity of $V$. 

Having reconstructed $K$, the matrix $A$ in reduced row echelon form is determined uniquely by its nullspace. Group the coordinates $x_1,\ldots,x_n$ into exactly $r$ pivot coordinates $x_p$ and $n-r$ free coordinates $x_u$. Setting $x_u=0$ in \eqref{eq:cond_normalform} results in the simplified condition $x_p \eq \eta$. It remains to show that we can recover the observing distribution $\mu$ from observational behavior. Intuitively, if we put a flat enough prior on $x_p$, the posterior will resemble $\mu$ arbitrarily closely:

Let $\mu=\mathcal N(b,\Sigma)$ and consider $x_p \sim \mathcal N(\lambda I)$ for $\lambda \to \infty$. The matrix $(I + \lambda^{-1}\Sigma)$ is invertible for all large enough $\lambda$. By the formulas \ref{prop:gauss_observe}, the mean of the posterior is
\[ \bar \mu = (I+\lambda^{-1}\Sigma)^{-1}\mu \xrightarrow{\lambda \to \infty} \mu \]
For the covariance, we truncate Neumann's series 
\[ (I + \lambda^{-1}\Sigma)^{-1} = I - \lambda^{-1}\Sigma + o(\lambda^{-2}) \]
to obtain
\[ \bar \Sigma = \lambda I - \lambda(I + \lambda^{-1}\Sigma)^{-1} = \Sigma + o(\lambda^{-2}) \xrightarrow{\lambda \to \infty} \Sigma \]
\end{proof}

\subsection{Implementation}
\label{sec:impl}

It is straightforward to implement the operational semantics of \S\ref{sec:gaussian_language} in a language like Python. We have done this and we illustrate with further simple programs and results, in addition to the examples in Section~\ref{sec:introA}.

\begin{lstlisting}[caption=Gaussian regression (Fig.~\ref{fig:ridge})]
  xs = [1.0, 2.0, 2.25, 5.0, 10.0]
  ys = [-3.5, -6.4, -4.0, -8.1, -11.0]
  
  a = Gauss.N(0,10)
  b = Gauss.N(0,10)
  f = lambda x: a*x + b
  
  for (x,y) in zip(xs,ys):
    Gauss.condition(f(x), y + Gauss.N(0,0.1))
  \end{lstlisting}
  \begin{figure}[h]
    \centering
	\includegraphics[width=0.5\textwidth]{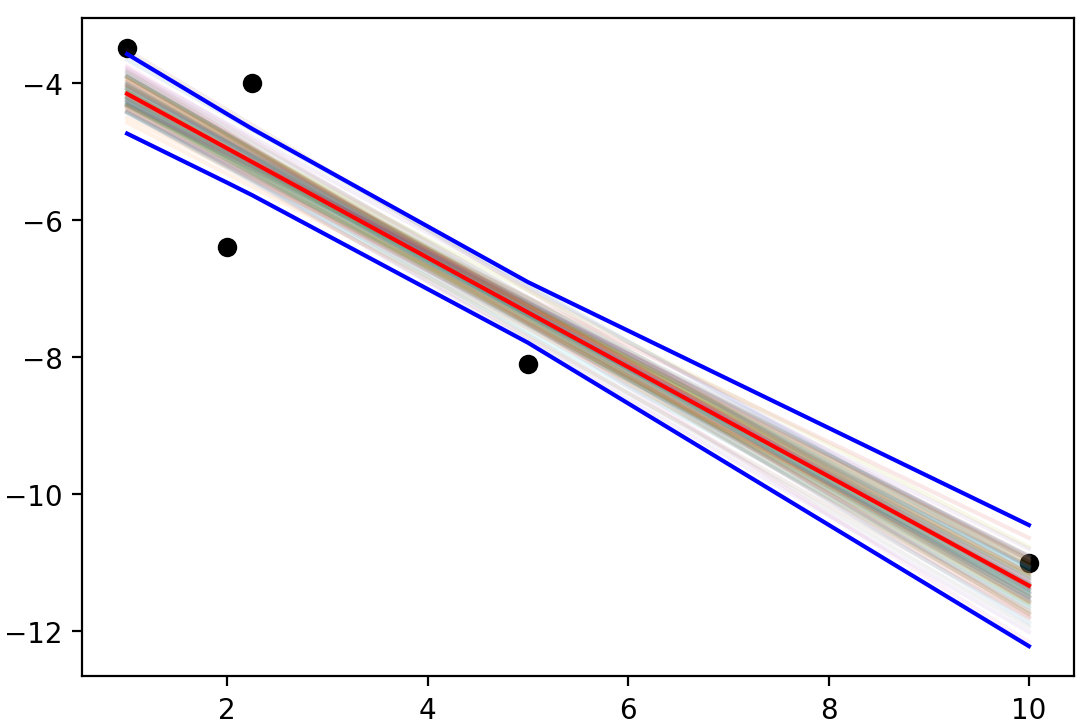}
	\caption{Gaussian regularized regression (Ridge regression), plotting 100 samples, the mean (red) and $\pm 3\sigma$ coordinatewise (blue) \label{fig:ridge}}
\end{figure}

\begin{figure}[h]
	\begin{lstlisting}[caption=1-dimensional K\'{a}lm\'{a}n filter (Fig.~\ref{fig:kalman})]
    xs = [ 1.0, 3.4, 2.7, 3.2, 5.8, 
           14.0, 18.0, 11.7, 19.5, 19.2]
    
    x = [0] * len(xs)
    v = [0] * len(xs)
    
    # Initial parameters
    x[0] = xs[0] + Gauss.N(0,1)
    v[0] = 1.0 + Gauss.N(0,10)
    
    for i in range(1,len(xs)):
      # Predict movement 
      x[i] = x[i-1] + v[i-1]
      v[i] = v[i-1] + Gauss.N(0,0.75)
      
      # Make noisy observations
      Gauss.condition(x[i] + Gauss.N(0,1),xs[i])
    \end{lstlisting}\centering
	\includegraphics[width=0.5\textwidth]{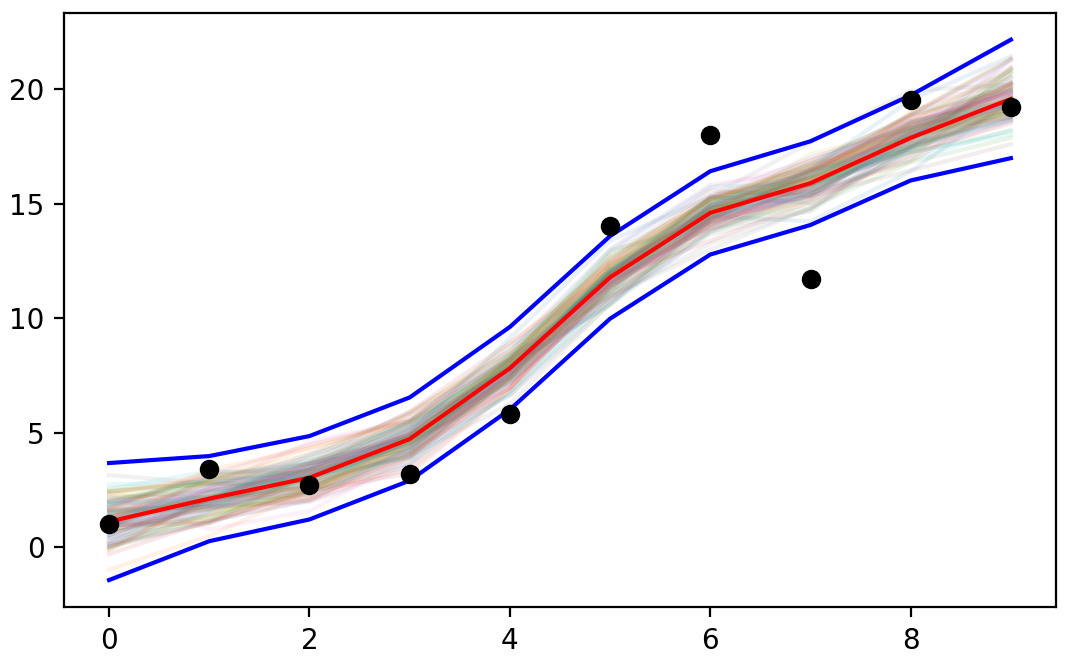}
	\caption{K\'{a}lm\'{a}n filter example\label{fig:kalman}}
\end{figure}

\end{document}